\documentclass[journal,10pt]{IEEEtran}
\ifCLASSINFOpdf
\else
\fi
\usepackage{comment} 

\usepackage{amsmath,graphicx}

\usepackage{pstricks,graphicx}
\usepackage{setspace}
\usepackage{psfrag}
\usepackage{amssymb}
\usepackage{amsmath}
\usepackage{amsthm}

\newtheorem{remark}{Remark}

\newtheorem{lemma}{Lemma}

\newtheorem{theorem}{Theorem}

\usepackage{pstool}

\usepackage{multirow}

 \usepackage{algorithm}
 \usepackage{algorithmic}
 
\begin{document}

%\onecolumn

%\doublespacing 

\title{Revisiting the Primal-Dual Method of Multipliers for  Optimisation over Centralised Networks}

\graphicspath{{figures/}}

\author{Guoqiang~Zhang, Kenta Niwa and W. Bastiaan Kleijn
\thanks{G.~Zhang is with the School of Electrical and Data Engineering,  University of Technology, Sydney, Australia. Email: {guoqiang.zhang@uts.edu.au}}

\thanks{K. Niwa is with both Communication Science Laboratories and Computer and Data Science Laboratories, Nippon Telegraph and Telephone Corporation (NTT).
Email: {kenta.niwa.bk@hco.ntt.co.jp}}

\thanks{W.~Bastiaan~Kleijn is Victory University of Wellington.
Email: {bastiaan.kleijn@ecs.vuw.ac.nz}}
}

\maketitle

\begin{abstract}
The primal-dual method of multipliers (PDMM) was originally designed for solving a decomposable optimisation problem over a general network. In this paper, we revisit PDMM for optimisation over a centralized network. We first note that the recently proposed method \emph{FedSplit}  \cite{Pathak2021} implements PDMM for a centralized network. In \cite{Pathak2021},  Inexact FedSplit (i.e., gradient based FedSplit) was also studied both empirically and theoretically.  We identify the cause for the poor reported performance of Inexact FedSplit, which is due to the improper initialisation in the gradient operations at the client side. To fix the issue of Inexact FedSplit, we propose two versions of Inexact PDMM, which are referred to as gradient-based PDMM (GPDMM) and accelerated GPDMM (AGPDMM), respectively. AGPDMM accelerates GPDMM at the cost of transmitting two times the number of parameters from the server to each client per iteration compared to GPDMM. We provide a new convergence bound for GPDMM for a class of convex optimisation problems. Our new bounds are tighter than those derived for Inexact FedSplit. We also investigate the update expressions of AGPDMM and SCAFFOLD to find their similarities. It is found that when the number $K$ of gradient steps at the client side per iteration is $K=1$, both AGPDMM and SCAFFOLD reduce to vanilla gradient descent with proper parameter setup. Experimental results indicate that AGPDMM converges faster than SCAFFOLD when $K>1$ while GPDMM converges slightly worse than SCAFFOLD. 
 
%We then consider applying the new algorithm for distributed averaging. For the case of no transmission failure, the new algorithm remarkably outperforms the state-of-the-art methods. For the case of transmission losses, the new algorithm is robust to transmission-failure.
\end{abstract}



% \fi
\IEEEpeerreviewmaketitle

\vspace{-4mm}
\section{Introduction}
\vspace{-2mm}
In the last decade, distributed optimisation \cite{Boyd11ADMM} has drawn increasing attention due to the demand for massive-data processing and easy remote access to ubiquitous computing units (e.g., a computer or a mobile phone) over a network.  Its basic principle is to allocate the data over a set of computing units instead of one server and then allow the computing units to collaborate with each other in a distributed manner to iteratively obtain a global solution (e.g., a machine learning (ML) model) of an optimisation problem which is formulated via the data.  In general, the typical challenges faced by distributed optimisation include, for instance, data-heterogeneity across the network,  expensive communication, data-privacy requirements, massive scalability,  and heterogeneous local computational resources \cite{Dimakis10GossipAlg, Li19Fed}. Depending on the applications, various methods have been developed for addressing one or more challenges in the considered network (e.g., \cite{Boyd06gossip, Richtarik16Dist, Li14Dist}). 

%By doing so,  distributed optimisation facilitates the protection of personal data which is expected to be the future trend in data analytics \cite{GDPR19}. 

Considering the application of distributed optimisation for learning an ML model, distributed learning \cite{Zhang16PDMM, Kenta20KDD} over a decentralized (i.e., peer-to-peer (P2P)) network and federated learning \cite{Kairouz19Fed} over a centralised (i.e., server-client topology) network have been two of the most active research topics in recent years. In a P2P network, network nodes can be connected arbitrarily in an equal relationship. In this situation, distributed optimisation methods are designed to be node-independent w.r.t. local computation and communication to  enable network scalability.  The algorithms in the literature can be roughly classified as either average-consensus based or primal-dual based.  

In brief, the average-consensus approach  \cite{Li20gradientTrack, Li19GradTrack, Blot16Gossip} allows the network nodes to share and average (or fuse) the estimated models to be learned among neighbours iteratively until reaching global consensus.  On the other hand, the primal-dual approach \cite{Zhang16PDMM, Kenta20KDD, Rajawat20PDMM} intends to explicitly represent the neighbouring consensus requirements via linear equality constraints in terms of neighbouring model variables and then iteratively solve the reformulated optimisation problem via either Peaceman-Rachford  (PR) splitting or Douglas-Rachford (DR) splitting (e.g., \cite{OksendalBook03, Ryu16Mono}). In particular, the alternating direction method of multipliers (ADMM) \cite{Giselsson17ADMM}  and the primal-dual method of multipliers (PDMM) \cite{Zhang16PDMM,Sherson17PDMM} are two known algorithms based on DR splitting and PR splitting, respectively. One major advantage of the second approach is that it is able to handle heterogeneous\footnote{Alternatively referred to as non i.i.d. data across different network nodes. } data implictly by imposing linear equality constraints w.r.t. model variables. %, for which more details will be provided later on.  

Federated learning focuses on networks with server-client topologies  \cite{Kairouz19Fed}. In the learning procedure, the server is responsible for collecting, fusing, and broadcasting information from/to all the clients while each client only needs to communicate with the server directly, which makes it easily implementable. In general, federated learning is more time-effective through global information collection and spread at the cost of limited scalability than distributed learning over a P2P network \cite{Li19Fed}. The algorithms developed for a P2P network (e.g.,  \cite{Kenta20KDD}) can often be utilised for federated learning by viewing the server-client structure as a special type of P2P network. Recent developed algorithms for federated learning include, for example, FEDAC \cite{Yuan20Fed}, FedSplit \cite{Pathak2021},  and SCAFFOLD \cite{Karimireddy20SCAFFOLD}. SCAFFOLD can be viewed as belonging to the primal-dual approach due to the introduced covariates in its update expressions for compensating the functional heterogeneity. 

In this paper, we revisit the primal-dual method of multipliers (PDMM) proposed in \cite{Zhang16PDMM, Connor17PDMM}. The method was originally designed to solve a decomposable optimisation problem over a graphical model $\mathcal{G}=(\mathcal{V},\mathcal{E})$:
\begin{align}
\hspace{-3mm}\min_{\{\boldsymbol{x}_i\}}\hspace{-0.6mm}\Big(\hspace{-0.3mm} \sum_{i\in \mathcal{V}} f_i(\boldsymbol{x}_i) \hspace{-0.3mm}\Big) \; \textrm{s. t.} \; \boldsymbol{B}_{i|j} \boldsymbol{x}_i \hspace{-0.6mm}=\hspace{-0.6mm}  \boldsymbol{B}_{j|i}  \boldsymbol{x}_j \; \forall (i,j)\in \mathcal{E}, 
\label{equ:optiGeneral}
\end{align}
where the notation $\textrm{s.~t.}$ stands for ``subject to", $\mathcal{V}$ and $\mathcal{E}$ represent the sets of nodes and undirected edges respectively, and $f_i(\cdot)$ denotes the local function at node $i\in \mathcal{V}$. The two constant matrices $\boldsymbol{B}_{i|j}$ and $\boldsymbol{B}_{j|i}$ specify the linear equality constraint for $(i,j)\in \mathcal{E}$.  As PDMM belongs to PR splitting, it enjoys the benefit that PR splitting gives the best convergence bounds with proper parameter setups for a certain class of functions \cite[Remark 4]{Giselsson17ADMM}. The recent work \cite{Kenta20KDD} has successfully applied Inexact PDMM (or gradient based PDMM) for training deep neural networks (DNNs) over P2P networks to the case of heterogeneous data.  In \cite{Rajawat20PDMM}, the authors successfully extend PDMM by incorporating SAGA, L-SVRG, and SVRG++ over P2P networks. The performance of PDMM for centralised networks remains to be explored. 

%Our main contribution in this 

This paper studies the relationship between PDMM, and the two methods FedSplit and SCAFFOLD from the literature for optimisation over centralised networks. Our contributions are three-fold.  Firstly, it is found that PDMM reduces to FedSplit when applied to a centralized network.  We identify the cause for the poor reported performance of Inexact FedSplit (i.e., gradient based FedSplit) in \cite{Pathak2021},  as being due to the improper parameter initialisation at the client side per iteration.  

Secondly, to correct the issue of Inexact FedSplit, we propose two versions of inexact PDMM, which are referred to as gradient-based PDMM (GPDMM) and  accelerated GPDMM (AGPDMM), respectively. It is noted that GPDMM only needs to transmit one variable (a combination of a primal variable and a dual variable) per iteration between the server and clients. To accelerate the convergence speed of GPDMM,  AGPDMM is designed to transmit two variables (a primal variable and a dual variable) per iteration from the server to the clients.  Linear convergence rates for strongly convex and sublinear convergence rates for general convex cases are then established for GPDMM, which lead to tighter convergence bounds than those in \cite{Pathak2021}.  We note that, in principle, the analysis results in \cite{Connor17PDMM, Rajawat20PDMM} for GPDMM over a decentralied network also hold for centralised networks. However, \cite{Connor17PDMM} only shows the convergence of  GPDMM while the recent work \cite{Rajawat20PDMM} only shows the sublinear convergence rates. 

Thirdly, it is found that both AGPDMM and SCAFFOLD reduce to the vanilla gradient descent operation under proper parameter setup when the number $K$ of gradient steps at the client side per iteration is set to $K=1$. Experimental results show that GPDMM produces slightly worse performance than SCAFFOLD which transmits two variables between the server and clients per iteration. On the other hand, AGPDMM converges faster than SCAFFOLD when $K>1$.

\vspace{-3mm}
\section{Problem Description}
\vspace{-1mm}
%\subsection{Notations}
\noindent\textbf{Notation and definition of a convex conjugate function}:   We use bold small letters to denote vectors and bold capital letters to denote matrices.  In particular, $\boldsymbol{I}$ denotes the identity matrix. The superscript $(\cdot)^T$ represents the transpose operator. Given a vector $\boldsymbol{y}$, we use $\|\boldsymbol{y}\|$ to denote its $l_2$ norm. Given a graphical model $\mathcal{G}=(\mathcal{V}, \mathcal{E})$, we use $\mathcal{N}_i$ to denote the set of neighbours for node $i$.  Suppose $h:\mathbb{R}^n\rightarrow \mathbb{R}\cup \{+\infty\}$ is a closed, proper and convex function.  Then the conjugate of $h(\cdot)$ is defined as \cite{SawaragiBook85}[Definition 2.1.20]
\begin{align}
h^{\ast}(\boldsymbol{\delta})\stackrel{\Delta}{=} \max_{\boldsymbol{y}} \boldsymbol{\delta}^T\boldsymbol{y}-h(\boldsymbol{y}), \label{equ:conj_def}
\end{align}
where the conjugate function $h^{\ast}$ is again a closed, proper and convex function. 
%Let $\boldsymbol{y}'$ be the optimal solution for a particular $\boldsymbol{\delta}'$ in (\ref{equ:conj_def}). We then  have
%\begin{align}
%\boldsymbol{\delta}'\in \partial_{\boldsymbol{y}}h(\boldsymbol{y}'),  \label{equ:conj_def2} 
%\end{align}
%where $\partial_{\boldsymbol{y}}h(\boldsymbol{y}')$ represents the set of all subgradients of $h(\cdot)$ at $\boldsymbol{y}'$ (see \cite[Definition 2.1.23]{SawaragiBook85}).  
%As a consequence, since $h^{\ast\ast}=h$, we have 
%\begin{align}
%h(\boldsymbol{y}')=&\boldsymbol{y}'^T\boldsymbol{\delta}'-h^{\ast}(\boldsymbol{\delta}')=\max_{\boldsymbol{\delta}} \boldsymbol{y}'^T\boldsymbol{\delta}-h^{\ast}(\boldsymbol{\delta}), \label{equ:conj_def3} 
%\end{align}
%and we conclude that $\boldsymbol{y}'\in \partial_{\boldsymbol{\delta}}h^{\ast}(\boldsymbol{\delta}')$ as well. 

%\vspace{-2mm}
%\subsection{Problem assumption}
%\label{subsec:proAss}
%\vspace{-0mm}

\noindent\textbf{Problem settings}: As a special case of (\ref{equ:optiGeneral}), we focus on a network of one server responsible for coordinating the learning process of $m$ clients, which can be represented as  
\begin{align}
\hspace{-3mm}\min_{\{\boldsymbol{x}_s, \boldsymbol{x}_i\in\mathbb{R}^d \}}\hspace{-0.6mm}\left(\hspace{-0.3mm} \sum_{i=1}^m f_i(\boldsymbol{x}_i) \hspace{-0.3mm}\right) \; \textrm{s. t.} \;  \boldsymbol{x}_s \hspace{-0.6mm}=\hspace{-0.6mm}  \boldsymbol{x}_i \; i=1,\ldots, m,
\label{equ:optiFed}
\end{align}
where the edge set $\mathcal{E}$ in the graph is $\mathcal{E}=\{(i,s)\}_{i=1}^m$,  the server function $f_s(\boldsymbol{x}_s)=0$, and each client function $f_i: \mathbb{R}^{d}\rightarrow \mathbb{R}$ is both continuously differentiable with the Lipschitz continuous gradient $L>0$ \cite{Zhou18Duality}
\begin{align}
&\hspace{-3mm}f_i(\boldsymbol{y}_i) \geq f_i(\boldsymbol{x}_i) \hspace{-0.6mm}+\hspace{-0.6mm} \nabla f_i(\boldsymbol{x})^T(\boldsymbol{y}_i \hspace{-0.6mm}- \hspace{-0.6mm}\boldsymbol{x}_i)  \hspace{-0.6mm}\nonumber \\
&\hspace{8mm}+  \frac{1}{2L}\| \nabla f_i(\boldsymbol{x}_i)  \hspace{-0.7mm}- \hspace{-0.7mm} \nabla f_i(\boldsymbol{y}_i) \|^2, \label{equ:gradLips2} 
\end{align}
and (strongly) convex
\begin{align}
&\hspace{-3mm}f_i(\boldsymbol{y}_i) \geq f_i(\boldsymbol{x}_i) \hspace{-0.6mm}+\hspace{-0.6mm} \nabla f_i(\boldsymbol{x})^T(\boldsymbol{y}_i \hspace{-0.6mm}- \hspace{-0.6mm}\boldsymbol{x}_i)  \hspace{-0.6mm}+ \hspace{-0.6mm} \frac{\mu}{2}\| \boldsymbol{x}_i  \hspace{-0.7mm}- \hspace{-0.7mm} \boldsymbol{y}_i \|^2, \label{equ:muStrong} 
\end{align}
for all $ \boldsymbol{y}_i\in \mathbb{R}^d, \boldsymbol{x}_i \in \mathbb{R}^d$. It is noted that convergence analysis for GPDMM will be conducted for both strong convexity ($\mu>0$) and general convexity ($\mu=0$) later on.

It is worth noting that (\ref{equ:gradLips2}) is essential to prove the linear convergence speed of GPDMM later on. In principle, the gradient difference $\|\nabla f_i(\boldsymbol{x}_i) - \nabla f_i(\boldsymbol{y}_i)\|^2$ is able to capture how the estimates of the dual variables of the method evolve over iterations. %while the difference $\|\boldsymbol{x}_i \hspace{-0.5mm}-\hspace{-0.5mm} \boldsymbol{y}_i\|^2 $ in (\ref{equ:gradLips}) can reflect how the estimates of the primal variables of the method evolve over iterations. 

The Lagrangian function for (\ref{equ:optiFed}) can be constructed as
\begin{align}
\mathcal{L}(\boldsymbol{x}_s, \{\boldsymbol{x}_i, \boldsymbol{\delta}_i\})= \sum_{i=1}^m f_i(\boldsymbol{x}_i) + \sum_{i=1}^m\boldsymbol{\delta}_{i}(\boldsymbol{x}_s -\boldsymbol{x}_i), \label{equ:fed_Lag}
\end{align}
where $\{\boldsymbol{\delta}_i\}$ are the Lagrangian multipliers, and can also be viewed as the \emph{dual} variables as opposed to the primal variables $\boldsymbol{x}_s$ and $ \{\boldsymbol{x}_i\}$. We assume there exists a saddle point  $\boldsymbol{x}_s^{\star}, \{\boldsymbol{x}_i^{\star}, \boldsymbol{\delta}_i^{\star}\}$ for (\ref{equ:fed_Lag}). The corresponding KKT conditions are given by 
\begin{align}
\nabla f_i(\boldsymbol{x}_i^{\star}) = \boldsymbol{\delta}_i^{\star}\; \forall i,  \quad \boldsymbol{x}_i^{\star} = \boldsymbol{x}_s^{\star} \; \forall i, \quad \sum_{i=1}^m \boldsymbol{\delta}_i^{\star} = 0.   \label{equ:KKT3}
\end{align} 
%\begin{align}
%\nabla f_i(\boldsymbol{x}_i^{\ast}) = \boldsymbol{\delta}_i^{\ast} \label{equ:KKT1} \\
%\boldsymbol{x}_i^{\ast} = \boldsymbol{x}_s^{\ast}  \label{equ:KKT2} \\
%\sum_{i=1}^m \boldsymbol{\delta}_i^{\ast} = 0   \label{equ:KKT3}
%\end{align} 
The research goal is to obtain a good estimation of $\boldsymbol{x}_s^{\star}$ via local computation and communication between the server and the $m$ clients after a reasonably number of iterations. We will propose two versions of Inexact PDMM by inspection of the update expressions of PDMM later on to reduce the computational complexity of PDMM per iteration.

%It is immediate that  optimal solution $\boldsymbol{x}_s^{\star}$ to the above problem (\ref{equ:optiFed}). The research goal is to compute or obtain a good approximation of $\boldsymbol{x}_s^{\star}$ via local computation and communication between the sever and the $m$ clients after a reasonably number of iterations. %To achieve the goal, the main challenge is to decide what information should be sent from a node to its neighbours per iteration and how to make use of the received information at each node for local computation.    

%The KKT conditions (\ref{equ:KKT3}) indicate that the individual gradient $\nabla f_i(\boldsymbol{x}_i^{\star})$ does not need to be zero. Only the gradient summation has to be zero. This implies that in practice, the client functions $\{f_i\}$ could be built on heterogeneous data. 

\vspace{-2mm}
\section{Relationship between PDMM and FedSplit}
\vspace{-1mm}
\label{sec:PDMM_Fedsplit}
In this section, we first briefly describe the updating procedure of PDMM for both the general problem (\ref{equ:optiGeneral}) and the special case (\ref{equ:optiFed}).  We will then explain that the recently developed method \emph{FedSplit} is identical to PDMM for solving the special problem (\ref{equ:optiFed}). After that, the poor performance for Inexact FedSplit in \cite{Pathak2021} will be studied. 
 
\vspace{-2mm} 
\subsection{ PDMM}
\vspace{-1mm}
\noindent \textbf{Iterates over a general graph}: Before introducing the method, we first present the dual problem for (\ref{equ:optiGeneral}), which can be obtained by constructing and optimising the so-called (primal) Lagrangian function 
\begin{align}
\hspace{-1mm}&\hspace{0mm}\max_{ \{\boldsymbol{\delta}_{ij}\} }\min_{\{\boldsymbol{x}_i\}} \Big( \sum_{i\in \mathcal{V}}\hspace{-0.6mm} f_i(\boldsymbol{x}_i)\hspace{-0.6mm}-\hspace{-2mm}\sum_{(i,j)\in \mathcal{E}}\hspace{-1.5mm}\boldsymbol{\delta}_{ij}^{T}(\boldsymbol{B}_{i|j}\boldsymbol{x}_i\hspace{-0.6mm}-\hspace{-0.8mm}\boldsymbol{B}_{j|i}\boldsymbol{x}_j\hspace{-0.3mm}) \Big) \nonumber \\
\hspace{-1mm}&\stackrel{(a)}{\footnotesize \Longleftrightarrow} \hspace{-2mm} \max_{ \{\boldsymbol{\lambda}_{i|j}, \boldsymbol{\lambda}_{j|i}\} } \min_{\{\boldsymbol{x}_i\}} \hspace{-1mm} \sum_{i\in \mathcal{V}}\hspace{-1.3mm}  \Big( f_i(\boldsymbol{x}_i) \hspace{-0.6mm} -\hspace{-0.6mm}  \boldsymbol{x}_i^T \hspace{-1.2mm}\sum_{j\in \mathcal{N}_i}\hspace{-1.5mm} \boldsymbol{B}_{i|j}^T\boldsymbol{\lambda}_{i|j}\hspace{-0.6mm} \Big),  \hspace{-1.2mm}\; \left\{\hspace{-2mm} \begin{array}{l}  \boldsymbol{\lambda}_{i|j} \hspace{-0.7mm}=\hspace{-0.7mm} -\hspace{-0.7mm} \boldsymbol{\lambda}_{j|i} \\ \forall (i,j)\in \mathcal{E} \end{array}\right. %\boldsymbol{\lambda}_{i|j} =\boldsymbol{\lambda}_{j|i} \forall (i,j)\in \mathcal{E} 
 \nonumber \\
\hspace{-1mm}&\stackrel{(b)}{\footnotesize \Longleftrightarrow} \hspace{-1.5mm} \max_{ \{\boldsymbol{\lambda}_{i|j}, \boldsymbol{\lambda}_{j|i}\} } \sum_{i\in \mathcal{V}}\hspace{-0.6mm}  -f_i^{\ast} \Big(  \sum_{j\in \mathcal{N}_i}\boldsymbol{B}_{i|j}^T\boldsymbol{\lambda}_{i|j} \Big), \hspace{-1.2mm}\; \left\{\hspace{-2mm} \begin{array}{l}  \boldsymbol{\lambda}_{i|j} = -\boldsymbol{\lambda}_{j|i} \\ \forall (i,j)\in \mathcal{E} \end{array}\right. \hspace{-4mm} , %\boldsymbol{\lambda}_{i|j} =\boldsymbol{\lambda}_{j|i} \forall (i,j)\in \mathcal{E} 
 \label{equ:dualGen}
\end{align}
where $\boldsymbol{\delta}_{ij}$ is the Lagrangian multiplier (or the dual variable) for each constraint $\boldsymbol{B}_{i|j}\boldsymbol{x}_i=\boldsymbol{B}_{j|i}\boldsymbol{x}_j$, which by using the lifting technique \cite{Zhang16PDMM}, can be further replaced by two dual variables $(\boldsymbol{\lambda}_{i|j},\boldsymbol{\lambda}_{j|i})$ under the constraint  $\boldsymbol{\lambda}_{i|j}=-\boldsymbol{\lambda}_{j|i}$ in step $(a)$. The variable $\boldsymbol{\lambda}_{i|j}$ is owned by node $i$ and is related to neighbour $j$.  It is noted that $\mathcal{N}_i$ denotes the set of neighbours for node $i$. $f_i^{\ast}$ in step $(b)$ is the conjugate function of $f_i$ (see (\ref{equ:conj_def}) for the definition).  We use $\boldsymbol{\lambda}_i$ to denote the vector by concatenating all $\boldsymbol{\lambda}_{i|j}$, $j\in\mathcal{N}_i$. Finally, we let $\boldsymbol{\lambda}=[\boldsymbol{\lambda}_1^T,\ldots,\boldsymbol{\lambda}_{|\mathcal{V}|}^T]^T$ and $\boldsymbol{x}=[\boldsymbol{x}_1^T,\ldots,\boldsymbol{x}_{|\mathcal{V}|}^T]^T$, where the dimension of $\boldsymbol{\lambda}$ depends on the network topology.

Instead of solving the primal problem (\ref{equ:optiGeneral}) or the dual one (\ref{equ:dualGen}) separately, PDMM is designed to iteratively approach a saddle point of an augmented primal-dual Lagrangian function obtained by combining  (\ref{equ:optiGeneral}) and (\ref{equ:dualGen}) \cite{Zhang16PDMM}:
\begin{align}
\hspace{-4mm}\mathcal{L}_{\rho}(\boldsymbol{x},\boldsymbol{\lambda}&)=\sum_{i\in \mathcal{V}} \hspace{-0.6mm}\Big[f_i(\boldsymbol{x}_i)+\hspace{-1.5mm}\sum_{j\in \mathcal{N}_i}\hspace{-0.5mm}\boldsymbol{\lambda}_{j|i}^T(\boldsymbol{B}_{i|j}\boldsymbol{x}_i) \nonumber\\
&\hspace{-3mm}-\hspace{-0.5mm}f_i^{\ast}\Big(\sum_{j\in\mathcal{N}_i}\boldsymbol{B}_{i | j}^T\boldsymbol{\lambda}_{i|j}\Big)\hspace{-0.5mm}\Big]
\hspace{-0.5mm}+h_{\rho}(\boldsymbol{x})\hspace{-0.5mm}-\hspace{-0.5mm}g_{\rho}(\boldsymbol{\lambda})
\label{equ:PDLag2} 
\end{align}
where $h_{\rho}(\boldsymbol{x})$ and $g_{\rho}(\boldsymbol{\lambda})$ are defined as
\begin{align}
\hspace{-1mm}h_{\rho }(\boldsymbol{x})=&\hspace{-1mm}\sum_{(i,j)\in \mathcal{E}} \hspace{-0.5mm}\frac{\rho }{2}\left\|\boldsymbol{B}_{i | j}\boldsymbol{x}_{i}-\boldsymbol{B}_{j | i}\boldsymbol{x}_{j}\right\|^2\label{equ:quadFunP}\\
g_{\rho}(\boldsymbol{\lambda})=&\sum_{(i,j)\in \mathcal{E}}\frac{1}{2\rho }\left\|\boldsymbol{\lambda}_{i|j}+\boldsymbol{\lambda}_{j|i}\right\|^2\hspace{-1mm},
\label{equ:quadFunD}
\end{align}
where $\rho>0$. $\mathcal{L}_{\rho}$ is convex in $\boldsymbol{x}$ and concave in $\boldsymbol{\lambda}$. 

Synchronous PDMM  optimises $\mathcal{L}_{\rho}$ by updating $\boldsymbol{x}$ and $\boldsymbol{\lambda}$ simultaneously per iteration through node-oriented computation. At iteration $r$, each $i$ computes a new estimate $\boldsymbol{x}_i^{r+1}$  by locally solving a small-size optimisation problem based on the neighbouring estimates $\{\boldsymbol{x}^r_j | j \in \mathcal{N}_i \}$ and $\{\boldsymbol{\lambda}^r_{j|i} | j \in \mathcal{N}_i\}$ from the last iteration:
\begin{align}
\hspace{-0mm}\boldsymbol{x}_i^{r+1} =&\arg\min_{\boldsymbol{x}_i} \Big[f_i(\boldsymbol{x}_i)+ \sum_{j\in \mathcal{N}_i}(\boldsymbol{\lambda}_{j|i}^{r})^T\boldsymbol{B}_{i|j}\boldsymbol{x}_i \nonumber
\end{align}
\begin{align}
\hspace{-0mm}&\hspace{1mm}+ \sum_{j\in \mathcal{N}_i}\frac{\rho }{2}\| \boldsymbol{B}_{i|j}\boldsymbol{x}_i  \hspace{-0.6mm} - \hspace{-0.6mm}  \boldsymbol{B}_{j|i}\boldsymbol{x}_j^{r}  \|^2 \Big] \quad i\in \mathcal{V} \label{equ:x_update} 
\end{align}
%\begin{align}
%\hspace{-0mm}\boldsymbol{\lambda}_i^{k+1} =&\arg\min_{\boldsymbol{\lambda}_i} \Big[f_i^{\ast}\Big(\sum_{j\in\mathcal{N}_i}\boldsymbol{B}_{i | j}^T\boldsymbol{\lambda}_{i|j}\Big)- \sum_{j\in \mathcal{N}_i}\boldsymbol{\lambda}_{i|j}^{T}\boldsymbol{B}_{j|i}\boldsymbol{x}_j^k \nonumber \\
%\hspace{-0mm}&\hspace{5mm}+\sum_{j\in \mathcal{N}_i} \frac{1}{2\rho_{ij}}\| \boldsymbol{\lambda}_{i|j}  \hspace{-0.6mm} + \hspace{-0.6mm}  \boldsymbol{\lambda}_{j|i}^{k}  \|^2 \Big] \quad i\in \mathcal{V}. \label{equ:lambda_update} 
%\end{align}
In principle, each estimate $\boldsymbol{\lambda}_{i}^{r+1}$ can be obtained similarly by solving a small-size optimisation problem that involves the conjugate function $f_i^{\ast}$ from (\ref{equ:PDLag2}). It is shown in \cite{Zhang16PDMM} that once $\boldsymbol{x}_i^{r+1} $ is obtained, $\{\boldsymbol{\lambda}_{i}^{r+1}\}$ can be computed directly as: 
\begin{align}
\hspace{-2mm}\boldsymbol{\lambda}_{i|j}^{r+1} \hspace{-1mm}=& \rho (\boldsymbol{B}_{j|i}\boldsymbol{x}_j^r \hspace{-0.7mm}-\hspace{-0.7mm} \boldsymbol{B}_{i|j}\boldsymbol{x}_i^{r+1} ) 
\hspace{-0.7mm}-\hspace{-0.7mm} \boldsymbol{\lambda}_{j|i}^r \; i\in \mathcal{V}, j\in \mathcal{N}_i. \label{equ:lambda_update_2nd}  
\end{align}
One can also design an asynchronous updating procedure for PDMM, where the network nodes are activated asynchronously for parameter updating at different iterations (see \cite{Zhang16PDMM} for more details). 

%fsad

We note that the above description of $\mathcal{L}_{\rho}$ and the update expressions (\ref{equ:x_update})-(\ref{equ:lambda_update_2nd}) for PDMM builds a foundation for the convergence analysis later on. The general linear constraints $\{\boldsymbol{B}_{i|j} \boldsymbol{x}_i \hspace{-0.6mm}=\hspace{-0.6mm}  \boldsymbol{B}_{j|i}  \boldsymbol{x}_j\}$ in (\ref{equ:optiGeneral}) enable PDMM to cover a broader class of problems than those methods which only focus on the special constraints $\{\boldsymbol{x}_i \hspace{-0.6mm}=\hspace{-0.6mm} \boldsymbol{x}_j\}$.  Another nice property of PDMM is that two dual variables $(\boldsymbol{\lambda}_{i|j}, \boldsymbol{\lambda}_{j|i})$ are introduced per linear constraint, which makes the update expressions node-oriented, thus facilitating practical implementation. It is shown in \cite{Sherson17PDMM} that PDMM can be alternatively derived from the PR splitting by using monotone operator theory \cite{Ryu16Mono}.  

\noindent \textbf{Iterates over the server-client graph for (\ref{equ:optiFed})}: We now consider applying PDMM to the problem (\ref{equ:optiFed}) by setting $\boldsymbol{B}_{i|s} =\boldsymbol{B}_{s|i}=\boldsymbol{I}$ for all the edges $(i,s)\in \mathcal{E}$. Instead of performing synchronous updates, we let the server compute the estimates $(\boldsymbol{x}_s^{r+1},\{\boldsymbol{\lambda}_{s|i}^{r+1}\})$  only after receiving the estimates $\{\boldsymbol{x}_i^{r+1}, \boldsymbol{\lambda}_{i|s}^{r+1}\}$ from the clients at iteration $r$. That is, at iteration $r$, the server uses the most up-to-date estimates $\{\boldsymbol{x}_i^{r+1}, \boldsymbol{\lambda}_{i|s}^{r+1}\}$ from the clients instead of the old estimates  $\{\boldsymbol{x}_i^{r}, \boldsymbol{\lambda}_{i|s}^{r}\}$ in computing $(\boldsymbol{x}_s^{r+1},\{\boldsymbol{\lambda}_{s|i}^{r+1}\})$. By inspection of (\ref{equ:x_update})-(\ref{equ:lambda_update_2nd}), one can then derive  the following update expressions with a slight index modification: 
\begin{align}
&\hspace{-2mm}  \textrm{clients}\hspace{-1mm}\left\{ \hspace{-2mm}\begin{array}{l}
\hspace{-0mm}\boldsymbol{x}_i^{r+1} \hspace{-1mm}=\hspace{-1mm} \arg\min_{\boldsymbol{x}_i} \hspace{-1mm} \Big[f_i(\boldsymbol{x}_i) \hspace{-0.7mm}+\hspace{-0.7mm} \frac{\rho}{2}\|\boldsymbol{x}_i  \hspace{-0.7mm} - \hspace{-0.7mm} \boldsymbol{x}_s^{r} \hspace{-0.7mm} +\hspace{-0.7mm}  \boldsymbol{\lambda}_{s|i}^{r}/\rho  \|^2 \Big]  \\
\hspace{0mm}\boldsymbol{\lambda}_{i|s}^{r+1} = \rho (\boldsymbol{x}_s^r \hspace{-0.7mm}-\hspace{-0.7mm} \boldsymbol{x}_i^{r+1} ) 
\hspace{-0.7mm}-\hspace{-0.7mm} \boldsymbol{\lambda}_{s|i}^r \end{array}\right. \label{equ:client_update} \\
&\hspace{-1mm} \textrm{server} \hspace{-1mm}\left\{ \hspace{-2mm}\begin{array}{l}
\hspace{-0mm}\boldsymbol{x}_s^{r+1} \hspace{-0.7mm}=\hspace{-0.7mm} \frac{1}{m}\sum_{i=1}^m (\boldsymbol{x}_i^{r+1} \hspace{-0.7mm}-\hspace{-0.7mm} \boldsymbol{\lambda}_{i|s}^{r+1}/\rho )  \\
\hspace{0mm}\boldsymbol{\lambda}_{s|i}^{r+1} = \rho (\boldsymbol{x}_i^{r+1} \hspace{-0.7mm}-\hspace{-0.7mm} \boldsymbol{x}_s^{r+1} ) 
\hspace{-0.7mm}-\hspace{-0.7mm} \boldsymbol{\lambda}_{i|s}^{r+1}  \end{array}\right. \hspace{-2mm}, \label{equ:server_update}
\end{align}
where the computation for $\boldsymbol{x}_{s}^{r+1}$ uses the fact that $f_s(\boldsymbol{x}_s)=0$. 

Next we briefly discuss the variables that must be transmitted between the server and the clients per iteration for PDMM to work. It is noted from (\ref{equ:client_update}) that at iteration $r$, each client $i$ only needs the quantity $ \boldsymbol{x}_s^{r} \hspace{-0.7mm} -\hspace{-0.7mm}  \boldsymbol{\lambda}_{s|i}^{r}/\rho $ from the server for the computation of $(\boldsymbol{x}_i^{r+1}, \boldsymbol{\lambda}_{i|s}^{r+1})$.  Similarly, the server only needs the quantity $ \boldsymbol{x}_i^{r+1} \hspace{-0.7mm} -\hspace{-0.7mm}  \boldsymbol{\lambda}_{i|s}^{r+1}/\rho $ from client $i$ to update $\boldsymbol{x}_s$ and $\boldsymbol{\lambda}_{s|i}$. That is, both the server and the client  need only to transmit one variable to each other per iteration, where the variable is a combination of the primal and dual estimates.  

\vspace{-2mm}
\subsection{(Inexact) FedSplit}
\vspace{-1mm}
\noindent \textbf{Iterates procedure}: Recently, the authors of \cite{Pathak2021} applied Peaceman-Rachford splitting to solve the special problem (\ref{equ:optiFed}). The resulting update expressions at iteration $r$ can be summarised as follows:
\begin{align}
&\hspace{-1mm}  \textrm{clients}\hspace{-1mm}\left\{ \hspace{-2mm}\begin{array}{l}
\hspace{-0mm}\boldsymbol{x}_i^{r+1} \hspace{-1mm}=\hspace{-1mm} \arg\min \hspace{-1mm} \Big[f_i(\boldsymbol{x}_i) \hspace{-0.6mm}+\hspace{-0.6mm} \frac{1}{2\gamma}\|\boldsymbol{x}_i  \hspace{-0.6mm} - \hspace{-0.6mm} \boldsymbol{z}_{s|i}^r \|^2 \Big]  \\
\hspace{0mm} \boldsymbol{z}_{i|s}^{r+1} = 2\boldsymbol{x}_i^{r+1} - \boldsymbol{z}_{s|i}^{r}  \end{array}\right. \label{equ:client_update_split} \\
&\hspace{-1mm} \textrm{server} \hspace{-1mm}\left\{ \hspace{-2mm}\begin{array}{l}
\hspace{-0mm}\boldsymbol{x}_s^{r+1} \hspace{-0.7mm}=\hspace{-0.7mm} \frac{1}{m}\sum_{i=1}^m \boldsymbol{z}_{i|s}^{r+1} \\
\boldsymbol{z}_{s|i}^{r+1} = 2\boldsymbol{x}_s^{r+1} - \boldsymbol{z}_{i|s}^{r+1} \end{array} \right. \hspace{-2mm}, \label{equ:server_update_split}
\end{align}
where the parameter $\gamma>0$, and $\{\boldsymbol{z}_{i|s}, \boldsymbol{z}_{s|i}\}$ are the auxiliary variables introduced in FedSplit. It is noted again that the clients only need to send $\{\boldsymbol{z}_{i|s}\}$ to the server for parameter updating while the server only needs to send $\boldsymbol{z}_{s|i}$ to client $i$, which is in line with that of PDMM. 

\noindent \textbf{On equivalence between PDMM and FedSplit}: We now briefly show that the iterates (\ref{equ:client_update})-(\ref{equ:server_update}) of PDMM reduce to  (\ref{equ:client_update_split})-(\ref{equ:server_update_split}) by proper hyper-parameter setup and reformulation. Specifically, by letting $\rho={1/\gamma}$,  $\boldsymbol{z}_{i|s} = \boldsymbol{x}_{i} - \gamma \boldsymbol{\lambda}_{i|s} $, and $\boldsymbol{z}_{s|i} = \boldsymbol{x}_{s} - \gamma \boldsymbol{\lambda}_{s|i} $ in (\ref{equ:client_update})-(\ref{equ:server_update}), one can easily oberse that the resulting expressions are identical to (\ref{equ:client_update_split})-(\ref{equ:server_update_split}). The equivalence between PDMM and FedSplit is due to the fact that both methods are based on Peaceman-Rachford splitting (see \cite{Sherson17PDMM} for more details about PDMM). However, PDMM is more general than FedSplit since it can also be applied for decentralised networks.

\noindent \textbf{Inexact iterates}: In practice, it might be difficult or expensive to obtain a closed form solution for $\boldsymbol{x}_i^{r+1}$ in (\ref{equ:client_update_split}) due to the complexity of $f_i(\boldsymbol{x}_i)$. One common practice is to conduct an inexact computation based on gradient descent.

The authors of \cite{Pathak2021} considered simplifying the minimisation problem in (\ref{equ:client_update_split}) by performing $K$ steps of consecutive gradient descent operations for each client $i$ at iteration $r$ to obtain a sequence of $K$ estimates: $\{\boldsymbol{x}_i^{r,k}| k=1, \ldots, K\}$. By starting with $\boldsymbol{x}_i^{r, k=0}=  \boldsymbol{z}_{s|i}^r$,  the estimate $\boldsymbol{x}_i^{r, k+1}$ at step $k$ of iteration $r$ is computed as 
\begin{align}
&\boldsymbol{x}_i^{r, k+1} = \boldsymbol{x}_i^{r, k} - \eta \nabla h_i^{r}( \boldsymbol{x}_i^{r, k}) \;\; 0\leq k <K,   
 \label{equ:gradient_xi_fedsplit}
\end{align}
where $\eta$ is the stepsize, and the function $h_i^{r}(\boldsymbol{x}_i)$ at iteration $r$ is defined to be 
\begin{align}
&h_i^{r}( \boldsymbol{x}_i)  = f_i(\boldsymbol{x}_i) + \frac{1}{2\gamma}\| \boldsymbol{x}_i -\boldsymbol{z}_{s|i}^r \|^2.  \label{equ:f_i_approximate_fedsplit}
\end{align}

%compare_simulation

We note that the initialisation $\boldsymbol{x}_i^{r, k=0}=  \boldsymbol{z}_{s|i}^r$ for the set of $K$ steps within each iteration is not a good option, especially for finite $K$ or small $\rho$ value.  From the analysis  about equivalence on PDMM and FedSplit, we notice that  $\boldsymbol{x}_i^{r, k=0}=  \boldsymbol{z}_{s|i}^r  = \boldsymbol{x}_s^{r} - \boldsymbol{\lambda}_{s|i}^{r}/\rho$. That is, $\boldsymbol{z}_{s|i}^r$ is a combination of both the primal and dual variables. A good initialisation of $\boldsymbol{x}_i^{r, k=0}$ should not include the dual variable $ \boldsymbol{\lambda}_{s|i}^r$.  This is because in general, the optimal solution  $\boldsymbol{\lambda}_{s|i}^{\ast}$  of the dual variable $\boldsymbol{\lambda}_{s|i}$ is not zero. Even the special initialisation $\boldsymbol{\lambda}_{s|i}^{r=0} = 0$ would not guarantee that $\boldsymbol{\lambda}_{s|i}^r$ is zero when the iteration $r>0$.  The component $ \boldsymbol{\lambda}_{s|i}^{r}/\rho$ makes the initialisation $\boldsymbol{x}_i^{r, k=0}  = \boldsymbol{x}_s^{r} - \boldsymbol{\lambda}_{s|i}^{r}/\rho$ less effective than an initialisation without the dual variable.  Small $\rho$ value would increase the impact of $\boldsymbol{\lambda}_{s|i}^{r}$. There are different ways to correct the improper initialisation of Inexact FedSplit depending on how to choose the estimates for $\{\boldsymbol{x}_i^{r, k=0}\}$. See the next section for the two versions of Inexact PDMM. 

%This suggests that even if $\boldsymbol{x}_i^{r, k=0}$ is initialised  to be $\boldsymbol{x}_i^{r, k=0}=  \boldsymbol{z}_{s|i}^{\ast}  = \boldsymbol{x}_s^{\ast} - \boldsymbol{\lambda}_{s|i}^{\ast}/\rho $ and $\boldsymbol{z}_{s|i}^r $ in (\ref{equ:f_i_approximate_fedsplit}) is set to $\boldsymbol{z}_{s|i}^{\ast} $, the new estimate $\boldsymbol{x}_i^{r,1}$ at $k=1$ would not be the same as $ \boldsymbol{x}_s^{\ast}$. A good initialisation of $\boldsymbol{x}_i^{r, k=0}$ should not include the dual variable $ \boldsymbol{\lambda}_{s|i}$. 

\begin{figure}[t!]
\centering
\includegraphics[width=60mm]{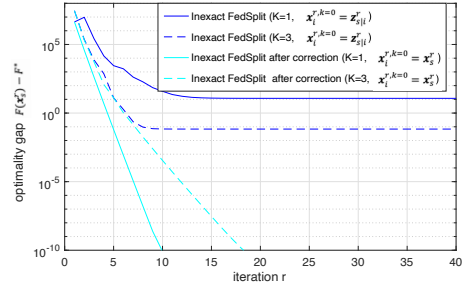}
\vspace*{-0.2cm}
\caption{\footnotesize{ Plots of the optimality gap $F(\boldsymbol{x}_s^r) - F^{\ast}$ versus the iteration number $r$ for Inexact FedSplit applied to a least-square problem over a network of 25 clients and one server, where $F(\boldsymbol{x}_s^r) =\sum_{i=1}^m f_i(\boldsymbol{x}_s^r)$ and $F^{\ast}$ denotes the minimum functional value.  See Subsection~\ref{subsec:least_square} for more details about the problem.  }}
\label{fig:FedSplit}
\vspace{-0.4cm}
\end{figure}

A simple evaluation of Inexact FedSplit is conducted for solving a least-square problem. As shown in Fig.~\ref{fig:FedSplit}, when the step number $K$ is finite (e.g., $K=1, 3$), Inexact FedSplit does not converge to the optimal solution due to the improper initialisation $\boldsymbol{x}_i^{r, k=0}=  \boldsymbol{z}_{s|i}^r$. If on the other hand, client $i$ initialises $\boldsymbol{x}_{i}^{r, k=0}$ to be $\boldsymbol{x}_{i}^{r, k=0}= \boldsymbol{x}_s^{r}$ at each iteration $r$, the method converges for both $K=(1,3)$.

%A simple evaluation on a least-square problem is conducted to understand the impact of the above two implementation differences in Inexact PDMM and Inexact FedSplit.  Fig.~\ref{fig:FedSplit_PDMM} displays the performance comparison. In addition to the implementation  (\ref{equ:gradient_xi})-(\ref{equ:f_i_approximate}), we also test the special initialisation $\boldsymbol{x}_i^{k, r=0}=\boldsymbol{z}_{s|i}^k=\boldsymbol{x}_s^k - \boldsymbol{\lambda}_{s|i}^k/\rho_i$ for Inexact PDMM. It is clear from the figure that when the step number $R$ is relatively small (e.g., $R=1, 3$), Inexact FedSplit is not able to converge to the optimal solution. One also observes that the special initialisation $\boldsymbol{x}_i^{k, r=0}=\boldsymbol{x}_s^k - \boldsymbol{\lambda}_{s|i}^k/\rho_i$ makes Inexact PDMM perform significantly worse than with the initialisation $\boldsymbol{x}_i^{k, r=0}=\boldsymbol{x}_i^k$. This might be because the special initialisation only considers the primal variable $\boldsymbol{x}_i$ and implictly ignores the dual variable $\boldsymbol{\lambda}_{i|s}$, which might break down the inherent information flow between the server and the clients in Inexact PDMM. Finally, Fig.~\ref{fig:FedSplit_PDMM} demonstrates that Inexact PDMM exhibits linear convergence rates. We will show later that for a well-defined least square problem, Inexact PDMM indeed converges in a linear rate. % for $R=1$. 
  
\textbf{Convergence bounds of Inexact FedSplit}: We note that the convergence bounds derived in \cite{Pathak2021} for Inexact FedSplit are not tight. Suppose all the client functions are strongly convex and have Lipschitz continuous gradients. Assume that at each iteration $r$, the error $\|\boldsymbol{x}_i^{r, k=K} - \boldsymbol{x}_i^{r, k=\infty}\|$ for each client is always upper-bounded by a scalar $b$. With proper setup for $\gamma$ in (\ref{equ:client_update_split}) and (\ref{equ:gradient_xi_fedsplit}), it is shown in  \cite{Pathak2021} that the error $\|\boldsymbol{x}_s^{r+1} - \boldsymbol{x}_s^{\star} \|$, $r\geq 1$,  is upper bounded by 
\begin{align}
&\|\boldsymbol{x}_s^{r+1} \hspace{-0.7mm} - \hspace{-0.7mm} \boldsymbol{x}_s^{\star} \| \leq \left(1 \hspace{-0.7mm}- \hspace{-0.7mm} \frac{2}{\sqrt{\kappa} +1}\right)^r  \hspace{-0.7mm} \frac{\| \boldsymbol{x}_s^{0}   \hspace{-0.7mm}- \hspace{-0.7mm} \boldsymbol{x}_s^{\star}  \|}{\sqrt{m}}  \hspace{-0.7mm}+ \hspace{-0.7mm} (\sqrt{\kappa}+1)b, \nonumber %\label{equ:Fedsplit_bound}
\end{align}
where the parameter $\kappa>0$ is determined by the properties (e.g., $L$, $\mu$ in (\ref{equ:gradLips2})-(\ref{equ:muStrong})) of the client functions. It is clear that the scalar $b$ is a loose offset for quantifying the error introduced by the gradient descent operations in Inexact FedSplit.  The convergence results in Fig.~\ref{fig:FedSplit} indicates that Inexact FedSplit may even not converge for small $K$, which can be explained by a large offset $b$.  %and therefore cannot reveal the real convergence speed. %Similarly, the convergence bound for the general convex functions $\{f_i\}$ in \cite{Pathak2021} is also loose. 

%it is not b that  bounds (reveals) the convergence speed but the above equation. It just says there is an offset b, which is precisely what their convergence curves show and relates to the problem. This bound suggests (but does not prove) that the iterates may never converge to the correct value.That could be said in a clearer way.

\vspace{-3mm}
\section{Inexact PDMM and its comparison to SCAFFOLD}
\vspace{-1mm}
\label{sec:GPDMM}
In this section, we first present the two versions of Inexact PDMM: namely, GPDMM and AGPDMM. In particular, GPDMM is designed for both the server and clients to transmit one variable to each other per iteration. To accelerate the convergence speed of GPDMM, AGPDMM requires the server to transmit two variables to each client per iteration. After that, we investigate the similarity of AGPDMM and SCAFFOLD. We show that when the number $K$ of gradient steps at the client slide per iteration is set to $K=1$, both AGPDMM and SCAFFOLD reduce to vanilla gradient descent under proper parameter setups. As will be discussed later, SCAFFOLD requires both the server and clients to transmit two variables to each other per iteration.

\vspace{-3mm}
\subsection{GPDMM by sending one variable from server to each client}
\vspace{-1mm}

\begin{algorithm}[tb]
   \caption{ GPDMM for a centralised network}
   \label{GCFOLD}
\begin{algorithmic}[1]
   \STATE {\bfseries Init.:}$\{\boldsymbol{x}_i^{r=0, K}\hspace{-0.7mm}=\hspace{-0.7mm}\boldsymbol{x}_s^{1}\}$, $\{\hspace{-0.7mm}\boldsymbol{\lambda}_{s|i}^{1}\hspace{-0.7mm}=\hspace{-0.7mm}0\}$, $\eta$, $\rho=\frac{1}{K\eta}$  
   \STATE For each iteration $r=1,\ldots, R$ do 
   \STATE $\;$ Server $s$ transmits $\boldsymbol{x}_s^{r} - \boldsymbol{\lambda}_{s|i}^r/\rho$ to each client $i$ 
   \STATE $\;$ On client $i$ in parallel do
    \STATE $\;\;\;$ Init.: $\boldsymbol{x}_i^{r,k=0} = \boldsymbol{x}_i^{r-1, K}$
   \STATE $\;\;\;$ For $k=0,\ldots, K-1$ do
   \STATE $\;\;\;\;\;$ $\boldsymbol{x}_i^{r, k+1} \hspace{-0.7mm}= \hspace{-0.7mm} \boldsymbol{x}_i^{r, k} \hspace{-0.8mm}-\hspace{-0.8mm} \frac{1}{1/\eta+\rho}  \big[\nabla f_i(\boldsymbol{x}_i^{r, k}) \hspace{-0.7mm}+\hspace{-0.7mm} \rho(\boldsymbol{x}_i^{r, k} \hspace{-0.7mm}-\hspace{-0.7mm} \boldsymbol{x}_s^{r}) \hspace{-0.7mm}+\hspace{-0.7mm} \boldsymbol{\lambda}_{s|i}^{r}\big] $ 
   \STATE $\;\;\;$ End for
    \STATE $\;\;\;$ $\boldsymbol{\lambda}_{i|s}^{r+1} = \rho (\boldsymbol{x}_s^{r} \hspace{-0.7mm}-\hspace{-0.7mm} \bar{\boldsymbol{x}}_i^{r, K} ) 
\hspace{-0.7mm}-\hspace{-0.7mm} \boldsymbol{\lambda}_{s|i}^{r} $ where $ \bar{\boldsymbol{x}}_i^{r, K}  \hspace{-0.7mm}=\hspace{-0.7mm} \frac{1}{K}\sum_{k=1}^K\boldsymbol{\lambda}_{i}^{r,k}$
    \STATE  $\;\;\;$ client $i$ transmits $ \bar{\boldsymbol{x}}_i^{r, K}  - \boldsymbol{\lambda}_{i|s}^{r+1}/\rho $ to server $s$ 
   \STATE $\;$ End on client
   \STATE $\;$  $\boldsymbol{x}_s^{r+1} \hspace{-0.7mm}= \hspace{-0.7mm} \frac{1}{m}\sum_{i=1}^m (\boldsymbol{x}_i^{r+1} \hspace{-0.7mm}-\hspace{-0.7mm} \boldsymbol{\lambda}_{i|s}^{r+1}/\rho ) $ 
   \STATE $\;$ $\boldsymbol{\lambda}_{s|i}^{r+1} = \rho ( \bar{\boldsymbol{x}}_i^{r, K}  \hspace{-0.7mm}-\hspace{-0.7mm} \boldsymbol{x}_s^{r+1} ) 
\hspace{-0.7mm}-\hspace{-0.7mm} \boldsymbol{\lambda}_{i|s}^{r+1} $
  \STATE End for
\end{algorithmic}
\end{algorithm}

To correct the convergence issue of Inexact FedSplit, GPDMM is designed to avoid using the estimate $\boldsymbol{x}_s- \boldsymbol{\lambda}_{s|i}/\rho$ when conducting approximate optimisation at the client side. Specifically, at iteration $r$, client $i$ sets $\boldsymbol{x}_i^{r, k=0}=\boldsymbol{x}_i^{r-1, K}$ and then performs $K$ steps of gradient-based approximate optimisations to obtain a sequence of estimates $\{\boldsymbol{x}_i^{r,1},\ldots, \boldsymbol{x}_i^{r,K}\}$.  The estimate $\boldsymbol{x}_i^{r,k+1}$  at step $k$ is computed as 
\begin{align}
\hspace{-2mm}\boldsymbol{x}_i^{r, k+1} \hspace{-1mm}&=\hspace{-1mm} \arg\min_{\boldsymbol{x}_i} \hspace{-1mm} \Big[f_i^{r,k}(\boldsymbol{x}_i) \hspace{-0.6mm}+\hspace{-0.6mm} \frac{\rho}{2}\|\boldsymbol{x}_i  \hspace{-0.6mm} - \hspace{-0.6mm} \boldsymbol{x}_s^{r} \hspace{-0.7mm} +\hspace{-0.7mm}  \boldsymbol{\lambda}_{s|i}^{r}/\rho  \|^2 \Big] \nonumber \\
\hspace{-2mm}&\hspace{-2mm}=\boldsymbol{x}_i^{r, k} \hspace{-0.7mm} - \hspace{-0.7mm} \frac{1}{1/\eta+\rho}\big[\nabla f_i(\boldsymbol{x}_i^{r, k}) \hspace{-0.7mm}+\hspace{-0.7mm} \rho(\boldsymbol{x}_i^{r, k} \hspace{-0.7mm}-\hspace{-0.7mm} \boldsymbol{x}_s^{r}) \hspace{-0.7mm}+\hspace{-0.7mm} \boldsymbol{\lambda}_{s|i}^r \big], \label{equ:gradient_xi}
\end{align}
where $f_i^{r,k}(\boldsymbol{x}_i)$ is a quadratic approximation  of $f_i(\boldsymbol{x}_i)$ at  $\boldsymbol{x}_i^{r,k}$:
\begin{align}
\hspace{-3mm}f_i^{r,k}(\boldsymbol{x}_i) \hspace{-0.1mm} =& \hspace{-0.1mm} f_i(\boldsymbol{x}_i^{r,k})  \hspace{-0.6mm}+\hspace{-0.6mm} (\boldsymbol{x}_i \hspace{-0.6mm}-\hspace{-0.6mm} \boldsymbol{x}_i^{r,k})^T \nabla f_i(\boldsymbol{x}_i^{r,k}) \hspace{-0.6mm} \nonumber \\
&+\hspace{-0.7mm} 1/(2\eta) \|\boldsymbol{x}_i \hspace{-0.6mm}-\hspace{-0.6mm} \boldsymbol{x}_i^{r,k} \|^2, \label{equ:f_i_approximate}
\end{align}
where $1/L\geq \eta>0$ is the gradient stepsize. The optimality condition for $\boldsymbol{x}_i^{r, k+1}$ in (\ref{equ:gradient_xi}) can be rewritten as
\begin{align}
\nabla f_i(\boldsymbol{x}_i^{r, k}) =& 1/\eta (\boldsymbol{x}_i^{r, k} \hspace{-0.6mm}-\hspace{-0.6mm} \boldsymbol{x}_i^{r,k+1} )  \nonumber \\
&- \rho (\boldsymbol{x}_i^{r,k+1}  \hspace{-0.6mm} - \hspace{-0.6mm} \boldsymbol{x}_s^{r} \hspace{-0.7mm} +\hspace{-0.7mm}  \boldsymbol{\lambda}_{s|i}^{r}/\rho).  
\label{equ:opti_r}
\end{align}
After finishing the computation for $\boldsymbol{x}_i^{r,K}$,  client $i$ then sets  $\boldsymbol{\lambda}_{i|s}^{r+1}$ to be 
\begin{align}
\boldsymbol{\lambda}_{i|s}^{r+1} \hspace{-0.6mm}=\hspace{-0.6mm} \rho\Big(\boldsymbol{x}_s^r \hspace{-0.6mm}-\hspace{-0.6mm}\frac{1}{K} \sum_{k=1}^K \boldsymbol{x}_i^{r,k}\Big)\hspace{-0.6mm}-\hspace{-0.6mm}\boldsymbol{\lambda}_{s|i}^{r}, \label{equ:lambda_update_client}
\end{align}
where, to facilitate convergence analysis, the average estimate $\frac{1}{K} \sum_{k=1}^K \boldsymbol{x}_i^{r,k}$ is used for computing $\boldsymbol{\lambda}_{i|s}^{r+1}$ instead of the final estimate $\boldsymbol{x}_i^{r,K}$. See remark below for our detailed motivation. 

\begin{remark}
\vspace{-2mm}
We note that the computation for $\boldsymbol{\lambda}_{i|s}^{r+1}$ in (\ref{equ:lambda_update_client}) is not the optimal setup from the viewpoint of  fast convergence speed.  One should replace the average estimate $\frac{1}{K} \sum_{k=1}^K \boldsymbol{x}_i^{r,k}$  in  (\ref{equ:lambda_update_client}) with the most recent estimate $\boldsymbol{x}_i^{r,K}$ when computing $\boldsymbol{\lambda}_{i|s}^{r+1}$, which can be represented as
\begin{align}
\boldsymbol{\lambda}_{i|s}^{r+1} \hspace{-0.6mm}=\hspace{-0.6mm} \rho\Big(\boldsymbol{x}_s^r \hspace{-0.6mm}-\hspace{-0.6mm} \boldsymbol{x}_i^{r,K}\Big)\hspace{-0.6mm}-\hspace{-0.6mm}\boldsymbol{\lambda}_{s|i}^{r}. \label{equ:lambda_GPDMM}
\end{align}
This is because the most recent estimate $\boldsymbol{x}_{i}^{r, K}$ provides a more accurate approximation of the optimal solution which minimises $f_i(\boldsymbol{x}_i) \hspace{-0.6mm}+\hspace{-0.6mm} \frac{\rho}{2}\|\boldsymbol{x}_i  \hspace{-0.6mm} - \hspace{-0.6mm} \boldsymbol{x}_s^{r} \hspace{-0.7mm} +\hspace{-0.7mm}  \boldsymbol{\lambda}_{s|i}^{r}/\rho  \|^2$ in (\ref{equ:client_update}) than the average estimate.
As will be analysed in next section, the average estimate $\frac{1}{K} \sum_{k=1}^K \boldsymbol{x}_i^{r,k}$ in (\ref{equ:lambda_update_client}) facilitates convergence analysis. We leave the convergence analysis for employing the update expression (\ref{equ:lambda_GPDMM}) for future research work. 
\vspace{-2mm}
\end{remark}

At the server side, once it receives the estimates $\{\boldsymbol{x}_i^{r+1} - \boldsymbol{\lambda}_{i|s}^{r+1}/\rho \}$ at iteration $r$, the estimates $\boldsymbol{x}_s^{r+1}$ and $\{\boldsymbol{\lambda}_{s|i}^{r+1}\}$ can be computed by following (\ref{equ:server_update}). By inspection of  (\ref{equ:server_update}), it is not difficult to show that 
\begin{align}
\sum_{i=1}^m\boldsymbol{\lambda}_{s|i}^{r+1} = 0, \label{equ:s_lambda_equality}
\end{align}
which always holds no matter how Inexact PDMM is performed at the client side. It is noted that the above equation is in line with one of the KKT conditions   in (\ref{equ:KKT3}). Equ.~(\ref{equ:s_lambda_equality}) will be used for convergence analysis later on.   See Alg.~1 for a brief summary for GPDMM, where $\rho$  is set to $\rho=1/(K\eta)$, which is inspired by the update expressions of SCAFFOLD as will be discussed later on.

%We denote the fixed point of (\ref{equ:client_update})-(\ref{equ:server_update}) to be $\boldsymbol{x}_s^{\star}$, $\{\boldsymbol{x}_{i}^{\star}\}_{i=1}^m$, and $\{\boldsymbol{\lambda}_{i|s}^{\star}, \boldsymbol{\lambda}_{s|i}^{\star}\}_{i=1}^m$. It can be shown that these quantities satisfy the following conditions
%\begin{align}
%sfd
%\end{align}

%It is also clear that PDMM can implictly handle data-heterogeneity. We use $\boldsymbol{x}_{i,local}^{\star}$ to denote the optimal solution of $f_i(\boldsymbol{x}_i)$, i.e., $\boldsymbol{x}_{i,local}^{\star}=\arg\min_{\boldsymbol{x}_i} f_i(\boldsymbol{x}_i)$. PDMM does not require these local optimal solutions $\{\boldsymbol{x}_{i,local}^{\star}\}$ to be identical. PDMM increase, the inconsistency between 
%\begin{align}
%&h_i^{r}( \boldsymbol{x}_i)  = f_i(\boldsymbol{x}_i) + \frac{1}{2\gamma}\| \boldsymbol{x}_i -\boldsymbol{z}_{s|i}^r \|^2.  \label{equ:f_i_approximate_fedsplit}
%\end{align}

There are two  differences between Inexact FedSplit and GPDMM. Firstly, each time, GPDMM approximates $f_i(\boldsymbol{x}_i)$ by (\ref{equ:f_i_approximate}) while  Inexact FedSplit  approximates the summation $h_i^{r}( \boldsymbol{x}_i)  = f_i(\boldsymbol{x}_i) + \frac{1}{2\gamma}\| \boldsymbol{x}_i -\boldsymbol{z}_{s|i}^r \|^2$ in (\ref{equ:f_i_approximate_fedsplit}) by a quadratic function. %Firstly, unlike Inexact FedSplit,  GPDMM takes the quadratic function $ \frac{\rho}{2}\| \boldsymbol{x}_i - \boldsymbol{x}_s^r + \boldsymbol{\lambda}_{s|i}^r/\rho  \|^2$ outside of the gradient operation. 
 Secondly,  Inexact FedSplit initialises $\boldsymbol{x}_i^{r, k=0}$ with the starting point $\boldsymbol{z}_{s|i}^r =  \boldsymbol{x}_s^r - \boldsymbol{\lambda}_{s|i}^r/\rho $ while GPDMM initialises $\boldsymbol{x}_i^{r, k=0}$ with the starting point $\boldsymbol{x}_i^{r-1, K}$ from the last iteration.  As concluded from last section, $\boldsymbol{z}_{s|i}$ involves both the primal and dual variables, and is thus not suitable for initialisation. 

\begin{algorithm}[tb]
   \caption{ AGPDMM for a centralised network}
   \label{GCFOLD}
\begin{algorithmic}[1]
   \STATE {\bfseries Init.:} $\boldsymbol{x}_s^{1}$, $\{\boldsymbol{\lambda}_{s|i}^{1}\hspace{-0.7mm}=\hspace{-0.7mm}0\}$, $\eta$, $\rho=\frac{1}{K\eta}$  
   \STATE For each iteration $r=1,\ldots, R$ do 
   \STATE $\;$ Server $s$ transmits $\boldsymbol{x}_s^r$ and $\boldsymbol{\lambda}_{s|i}^r$ to each client $i$ 
   \STATE $\;$ On client $i$ in parallel do
    \STATE $\;\;\;$ Init.: $\boldsymbol{x}_i^{r,k=0} = \boldsymbol{x}_s^{r}$
   \STATE $\;\;\;$ For $k=0,\ldots, K-1$ do
   \STATE $\;\;\;\;\;$ $\boldsymbol{x}_i^{r, k+1} \hspace{-0.7mm}= \hspace{-0.7mm} \boldsymbol{x}_i^{r, k} \hspace{-0.8mm}-\hspace{-0.8mm} \frac{1}{1/\eta+\rho}\big[\nabla f_i(\boldsymbol{x}_i^{r, k}) \hspace{-0.7mm}+\hspace{-0.7mm} \rho(\boldsymbol{x}_i^{r, k} \hspace{-0.7mm}-\hspace{-0.7mm} \boldsymbol{x}_s^{r}) \hspace{-0.7mm}+\hspace{-0.7mm} \boldsymbol{\lambda}_{s|i}^{r}\big] $ 
   \STATE $\;\;\;$ End for
    \STATE $\;\;\;$ $\boldsymbol{\lambda}_{i|s}^{r+1} = \rho (\boldsymbol{x}_s^{r} \hspace{-0.7mm}-\hspace{-0.7mm} \boldsymbol{x}_i^{r, K} ) 
\hspace{-0.7mm}-\hspace{-0.7mm} \boldsymbol{\lambda}_{s|i}^r$
    \STATE  $\;\;\;$ client $i$ transmits $\boldsymbol{x}_i^{r, K} - \boldsymbol{\lambda}_{i|s}^{r+1}/\rho $ to server $s$ 
   \STATE $\;$ End on client
   \STATE $\;$  $\boldsymbol{x}_s^{r+1} \hspace{-0.7mm}= \hspace{-0.7mm} \frac{1}{m}\sum_{i=1}^m (\boldsymbol{x}_i^{r, K} \hspace{-0.7mm}-\hspace{-0.7mm} \boldsymbol{\lambda}_{i|s}^{r+1}/\rho ) $ 
   \STATE $\;$ $\boldsymbol{\lambda}_{s|i}^{r+1} = \rho (\boldsymbol{x}_i^{r, K} \hspace{-0.7mm}-\hspace{-0.7mm} \boldsymbol{x}_s^{r+1} ) 
\hspace{-0.7mm}-\hspace{-0.7mm} \boldsymbol{\lambda}_{i|s}$
  \STATE End for
\end{algorithmic}
\vspace{-0.5mm}
\end{algorithm}

\subsection{AGPDMM by sending two variables from server to each client}
\label{subsec:AGPDMM}

\noindent \textbf{Updating and transmission procedure}: We note that the convergence speed of GPDMM can be accelerated by a slight modification of its updating expressions. It is known for both PDMM and GPDMM that the server aggregates information from all the clients at each iteration. At iteration $r$, the global estimate $\boldsymbol{x}_s^{r}$ should be  more accurate than each individual estimate $\boldsymbol{x}_i^{r-1, K}$. Therefore, it is preferable for each client $i$ to employ the global estimate $\boldsymbol{x}_s^{r}$ instead of $\boldsymbol{x}_i^{r-1, K}$  when conducting  $K$ steps of gradient-based approximate optimisation at iteration $r$.  That is, the quantity $\boldsymbol{x}_i^{r,k=0}$ should be initialised as $\boldsymbol{x}_i^{r,k=0} = \boldsymbol{x}_s^{r}$ to achieve fast convergence speed.  The computation for $\boldsymbol{\lambda}_{i|s}^{r+1}$ follows from (\ref{equ:lambda_GPDMM}) instead of (\ref{equ:lambda_update_client}) to further accelerate the convergence speed.  Alg.~2 summarises the updating procedure of AGPDMM, which is obtained by following the above guideline.  

We now briefly discuss the variables that need to be transmitted from the server to the clients. At iteration $r$, it is clear that AGPDMM has to send both $\boldsymbol{x}_s^{r}$ and  $\boldsymbol{\lambda}_{s|i}^{r}$ to each client $i$ to allow for parameter update while GPDMM only needs to send the combination $\boldsymbol{x}_s^{r}-\boldsymbol{\lambda}_{s|i}^{r}/\rho$ to client $i$. The two versions of inexact PDMM exhibit a trade-off between convergence speed and transmission bandwidth.  AGPDMM accelerates the convergence speed of GPDMM at the cost of transmitting two times the number of parameters as GPDMM from the server to each client per iteration. In practice, one can select a proper version of Inexact PDMM depending on the requirement of the considered application.

\noindent \textbf{Performance of AGPDMM when $K=1$}:  We will show in the following that under proper parameter selection, the update expression for AGPDMM when $K=1$ reduces to the vanilla gradient descent operation. Specifically, $\boldsymbol{x}_{s}^{r+1}$ at iteration $r$ can be represented as 
%\begin{align}
%&\hspace{-25mm}\boldsymbol{x}_{s}^{r+1}  = \frac{1}{m}\sum_{i=1}^m (\boldsymbol{x}_i^{r, K=1} \hspace{-0.7mm}-\hspace{-0.7mm} \boldsymbol{\lambda}_{i|s}^{r+1}/\rho )  \nonumber \\
% &\hspace{-18mm}\stackrel{(a)}{=} \frac{1}{m}\sum_{i=1}^m (2\boldsymbol{x}_i^{r, K=1} \hspace{-0.7mm} -  \boldsymbol{x}_s^{r} \hspace{-0.7mm} 
%\hspace{-0.7mm}+\hspace{-0.7mm} \boldsymbol{\lambda}_{s|i}^r/\rho )  \nonumber 
%\end{align}
\begin{align}
\hspace{-0mm}\boldsymbol{x}_{s}^{r+1}  &= \frac{1}{m}\sum_{i=1}^m (\boldsymbol{x}_i^{r, K=1} \hspace{-0.7mm}-\hspace{-0.7mm} \boldsymbol{\lambda}_{i|s}^{r+1}/\rho )  \nonumber \\
%&\hspace{-0mm}\stackrel{(a)}{=} \frac{1}{m}\sum_{i=1}^m (2\boldsymbol{x}_i^{r, K=1} \hspace{-0.7mm} -  \boldsymbol{x}_s^{r} \hspace{-0.7mm} 
%\hspace{-0.7mm}+\hspace{-0.7mm} \boldsymbol{\lambda}_{s|i}^r/\rho )  \nonumber  \\
 &\hspace{0mm}\stackrel{(a)}{=} \frac{1}{m}\sum_{i=1}^m (\boldsymbol{x}_s^{r} - \frac{2}{1/\eta+\rho}\big(\nabla f_i(\boldsymbol{x}_s^{r}) +\boldsymbol{\lambda}_{s|i}^{r}\big) \hspace{-0.7mm} +\hspace{-0.7mm} \boldsymbol{\lambda}_{s|i}^r/\rho )   \nonumber \\
 &\hspace{0mm}\stackrel{(b)}{=} \boldsymbol{x}_s^{r} -   \frac{2}{1/\eta+\rho} \frac{1}{m}\sum_{i=1}^m \nabla f_i(\boldsymbol{x}_s^{r})    \label{equ:xs_AGPDMM_K1_1}  \\ 
&\hspace{0mm}\stackrel{\rho=\frac{1}{\eta}}{=} \boldsymbol{x}_s^{r}  - \eta \frac{1}{m}\sum_{i=1}^m \nabla f_i(\boldsymbol{x}_s^{r}),  \label{equ:xs_AGPDMM_K1} 
\end{align}
where step $(a)$ utilises the expressions $\boldsymbol{\lambda}_{i|s}^{r+1} = \rho (\boldsymbol{x}_s^{r} \hspace{-0.7mm}-\hspace{-0.7mm} \boldsymbol{x}_i^{r, K=1} ) 
\hspace{-0.7mm}-\hspace{-0.7mm} \boldsymbol{\lambda}_{s|i}^r$ and $\boldsymbol{x}_i^{r, K=1}=\boldsymbol{x}_s^{r} - \frac{1}{1/\eta+\rho}\big[\nabla f_i(\boldsymbol{x}_s^{r}) +\boldsymbol{\lambda}_{s|i}^{r}\big]$. Step $(b)$ employs the equality (\ref{equ:s_lambda_equality}).  

It is clear from (\ref{equ:xs_AGPDMM_K1_1}) that the update expression for $\boldsymbol{x}_{s}^{r+1}$ is actually the vanilla gradient descent expression over the function $\frac{1}{m}\sum_{i=1}^m f_i(\boldsymbol{x})$ at the estimate $\boldsymbol{x}_s^r$. The estimates $\{\boldsymbol{\lambda}_{s|i}^r \}$ for the dual variables have no effect on the computation of $\boldsymbol{x}_{s}^{r+1}$. The parameter $\rho$ only affects the stepsize computation. When $\rho = \frac{1}{\eta}$,  the stepsize becomes $\eta$ as indicated by (\ref{equ:xs_AGPDMM_K1}).

\begin{remark}
Alternatively, we can take Inexact FedSplit with the special initialisation $\{\boldsymbol{x}_{i}^{r,k=0}=\boldsymbol{x}^r | r\geq 0\}$ as a variant of AGPDMM. In this case, one can show that the estimate $\boldsymbol{x}_s^{r+1}$ when $K=1$ is given by
\begin{align}
\hspace{-0mm}\boldsymbol{x}_{s}^{r+1}  & = \boldsymbol{x}_s^{r}  - 2\eta \frac{1}{m}\sum_{i=1}^m \nabla f_i(\boldsymbol{x}_s^{r}).
 \label{equ:xs_AGPDMM_var_K1} 
\end{align}
 It is seen that the step-size in (\ref{equ:xs_AGPDMM_var_K1}) is $2\eta$ in comparison to the step-size $\eta$ in (\ref{equ:xs_AGPDMM_K1}). This is because the quadratic term $\|\boldsymbol{x}_i -\boldsymbol{x}_s^{r} + \boldsymbol{\lambda}_{s|i}^{r+1}/\rho\|^2$ in (\ref{equ:client_update}) is treated differently in AGPDMM and its variant. 
\end{remark}
%any comment on what the effect of the dual variables is on the process? It accounts for the constraints, but what does that do? Does it resemble momentum methods in any way? Interesting that sign can be positive and negative.
\vspace{-3mm}
\subsection{Comparison with SCAFFOLD}
\vspace{-1mm}

\noindent \textbf{Updating and transmission procedure of SCAFFOLD}:  The recent  work \cite{Karimireddy20SCAFFOLD} proposes SCAFFOLD for stochastic distributed optimisation over a centralized network. To make a fair comparison with Inexact PDMM, we present the update expressions of SCAFFOLD for solving (\ref{equ:optiFed}), which can be represented as      
\begin{align}
&\hspace{-2mm}  \textrm{clients}\hspace{-1mm}\left\{ \hspace{-2mm}\begin{array}{l}
\hspace{0mm}\boldsymbol{x}_{i}^{r, 0} =  \boldsymbol{x}_{s}^{r}  \\
\hspace{-0mm}\boldsymbol{x}_i^{r, k+1} \hspace{-1mm}=\hspace{-1mm} \boldsymbol{x}_i^{r, k} \hspace{-1mm}-\hspace{-1mm} \eta (\nabla f_i(\boldsymbol{x}_i^{r, k}) \hspace{-1mm}-\hspace{-1mm}\boldsymbol{c}_{i}^r \hspace{-1mm}+\hspace{-1mm}\boldsymbol{c}^r\hspace{-0.6mm}) \;\; k \hspace{-0.5mm} =  |_{0}^{K-1}  \\
\hspace{0mm}\boldsymbol{c}_{i}^{r+1} =  \boldsymbol{c}_{i}^{r}  - \boldsymbol{c}^{r} +\frac{1}{K\eta} (\boldsymbol{x}_s^{r} -  \boldsymbol{x}_i^{r, K})  \end{array}\right. \hspace{-2.5mm} \label{equ:SCAFFOLD_client_update} \\
&\hspace{-1mm} \textrm{server} \hspace{-1mm}\left\{ \hspace{-2mm}\begin{array}{l}
\hspace{-0mm}\boldsymbol{x}_s^{r+1} \hspace{-0.7mm}=\hspace{-0.7mm} \boldsymbol{x}_s^{r} + \eta_g \frac{1}{m}\sum_{i=1}^m (\boldsymbol{x}_i^{r, K} - \boldsymbol{x}_s^{r})  \\
\hspace{0mm}\boldsymbol{c}^{r+1} = \boldsymbol{c}^{r}  +\hspace{-0.7mm} \frac{1}{m} \sum_{i=1}^m  (\boldsymbol{c}_{i}^{r+1}- \boldsymbol{c}_{i}^{r})   \end{array}\right., \hspace{-2mm} \label{equ:SCAFFOLD_server_update}
\end{align}
where all clients are included for information fusion at the server side per iteration, $k \hspace{-0.5mm} =  |_{0}^{K-1}$ is a short notation for $k=0,\ldots, K$, and $(\eta, \eta_g)$ are the stepsizes.  The parameters $\boldsymbol{c}$ and $\{\boldsymbol{c}_i\}$ are the so-called server and client control variates to compensate for the functional heterogeneity over different clients \cite{Karimireddy20SCAFFOLD}. From a high-level point of view, the control variates of SCAFFOLD play a similar role as the dual variables in (Inexact) PDMM. 

We point out that in the computation of $\boldsymbol{c}_{i}^{r+1}$ in (\ref{equ:SCAFFOLD_client_update}), the variable difference $(\boldsymbol{x}_s^{r} -  \boldsymbol{x}_i^{r, K})$ is scaled by the factor $\frac{1}{K\eta}$. In Alg.~1 and 2, the setup $\rho=\frac{1}{K\eta}$ is selected to ensure that the variable difference is also scaled by $\frac{1}{K\eta}$ in computing $\boldsymbol{\lambda}_{i|s}^{r+1}$.

From (\ref{equ:SCAFFOLD_client_update})-(\ref{equ:SCAFFOLD_server_update}), it is not difficult to conclude that at iteration $r$, the server needs to send the two variables $(\boldsymbol{x}_s^r, \boldsymbol{c}^r)$ to the clients to enable parameter update. Each client $i$ needs to send the two variables $(\boldsymbol{x}_i^{r, K} - \boldsymbol{x}_s^{r}, \boldsymbol{c}_{i}^{r+1}- \boldsymbol{c}_{i}^{r})$ to the server for information fusion. In contrast, the two versions of Inexact PDMM only require each client to transmit one variable to the server per iteration. The transmission load from the server to the clients depends on how Inexact PDMM is realised as discussed earlier. As will be shown in the experiment, AGPDMM converges faster than SCAFFOLD when $K>1$.

\noindent \textbf{Performance of SCAFFOLD when $K=1$}:  We now show that when $K=1$,  the update expression for $\boldsymbol{x}_s^{r+1}$ in (\ref{equ:SCAFFOLD_server_update}) also reduces to vanilla gradient descent operation under proper parameter selection.  Assume $\sum_{i=1}^m (\boldsymbol{c}_i^r- \boldsymbol{c}^r)\hspace{-0.6mm}=\hspace{-0.6mm}0$. It is immediate that 
\begin{align}
\boldsymbol{x}_s^{r+1} \hspace{-0.7mm}= \hspace{-0.7mm} \boldsymbol{x}_s^{r}  \hspace{-0.7mm} -\hspace{-0.7mm}   \frac{\eta_g\eta}{m} \hspace{-0.7mm} \sum_{i=1}^m \hspace{-0.7mm}  \nabla f_i(\boldsymbol{x}_s^{r}) \hspace{-0.7mm} \stackrel{\eta_g=1}{=}\hspace{-0.7mm} \boldsymbol{x}_s^{r}  \hspace{-0.7mm} -\hspace{-0.7mm}  \eta \frac{1}{m} \hspace{-0.7mm} \sum_{i=1}^m \hspace{-0.7mm} \nabla f_i(\boldsymbol{x}_s^{r}).\label{equ:xs_SCAFFOLD_K1_1}
\end{align}
One can also easily show that  $\sum_{i=1}^m (\boldsymbol{c}_i^{r+1}\hspace{-0.6mm}-\hspace{-0.6mm} \boldsymbol{c}^{r+1}) =0 $ based on the assumption $\sum_{i=1}^m (\boldsymbol{c}_i^r- \boldsymbol{c}^r)\hspace{-0.6mm}=\hspace{-0.6mm}0$.
Note that the parameter $\eta_g$ only affects the overall stepsize of the vanilla gradient descent. When $\eta_g=1$, (\ref{equ:xs_SCAFFOLD_K1_1})  is identical to (\ref{equ:xs_AGPDMM_K1}).

To summarise, when $K=1$, both SCAFFOLD and AGPDMM may reduce to the vanilla gradient descent operation. For SCAFFOLD, it is required that the initialisation $\sum_{i=1}^m (\boldsymbol{c}_i^0- \boldsymbol{c}^0)\hspace{-0.6mm}=\hspace{-0.6mm}0$. In the special case of $K=1$, the parameter $\rho$ in AGPDMM and $\eta_g$ in SCAFFOLD only affect the overall stepsizes of the vanilla gradient descent as discussed above. 

%Next we arThe quantity $\sum_{i=1}^m (\boldsymbol{c}_i^{r+1}\hspace{-0.6mm}-\hspace{-0.6mm} \boldsymbol{c}^{r+1}) $ can be represented as
%\begin{align}
%&\sum_{i=1}^m (\boldsymbol{c}_i^{r+1}- \boldsymbol{c}^{r+1}) \nonumber \\
%\end{align}

\vspace{-3mm}
\section{Convergence Analysis of GPDMM}
\vspace{-1mm}
\label{sec:convergenceAnalysis}

\noindent \textbf{An inequality for each estimate $\boldsymbol{x}_i^{r, k+1}$ }: Using the fact that the client functions $\{f_i\}$ are (strongly) convex and have Lipschitz continuous gradients,  we derive an inequality for $\boldsymbol{x}_i^{r, k+1}$ in (\ref{equ:gradient_xi}) at step $k$ of iteration $r$ in a lemma below: 
\begin{lemma}
Let $(1/\eta) \geq L$ in the approximation function (\ref{equ:f_i_approximate}). Then for any $\boldsymbol{x}_i\in \mathbb{R}^{d}$ and $\theta \in[0,1]$, we have
\begin{align}
& \hspace{-3mm} f_i(\boldsymbol{x}_i) -  f_i(\boldsymbol{x}_i^{r, k+1}) \nonumber \\
\geq &  \hspace{-0.6mm} (\boldsymbol{x}_i \hspace{-0.6mm}-\hspace{-0.6mm} \boldsymbol{x}_i^{r, k+1})^T  [ \rho(\hspace{-0.6mm} \boldsymbol{x}_s^{r} \hspace{-0.7mm} - \hspace{-0.6mm}  \boldsymbol{x}_i^{r,k+1} ) \hspace{-0.6mm} -\hspace{-0.6mm}  \boldsymbol{\lambda}_{s|i}^{r}] \hspace{-0.6mm}+\hspace{-0.6mm} \frac{1}{2\eta}  
\|\boldsymbol{x}_i - \boldsymbol{x}_i^{r, k+1} \|^2 \nonumber\\
& - \hspace{-0.6mm} \frac{1/\eta - \theta\mu}{2}\|\boldsymbol{x}_i^{r,k}-\boldsymbol{x}_i \|^2 \hspace{-0.6mm} +\hspace{-0.6mm} \frac{1/\eta- L}{2} \|\boldsymbol{x}_i^{r, k+1}  \hspace{-0.6mm}- \hspace{-0.6mm}\boldsymbol{x}_i^{r, k} \|^2 \nonumber \\
&+ \frac{1-\theta}{2L}\| \nabla f_i(\boldsymbol{x}_i^{r,k})- \nabla f_i ( \boldsymbol{x}_i) \|^2, \hspace{-2mm}
\label{equ:primal_inequality_general}
\end{align}
where $\mu=0$ corresponds to the general convex case.    %Basically, the inequality can be derived by using  (\ref{equ:opti_r}) and (\ref{equ:gradLips})-(\ref{equ:muStrong}). 
\label{lemma:primal_inequality_general}
\vspace{0mm}f
\end{lemma}
\begin{proof}
See Appendix~\ref{appendix:lemma_ineq} for detailed derivation. 
\end{proof}

\noindent\textbf{An inequality for all estimates $\{\boldsymbol{x}_i^{r, k}| k=1,\ldots, K \}_{i=1}^m$}:  %\begin{sloppypar}
Suppose $\{\boldsymbol{x}_s^{\star}=\boldsymbol{x}_i^{\star}\}_{i=1}^m$ together with $\{\boldsymbol{\lambda}_{i|s}^{\star} = - \boldsymbol{\lambda}_{s|i}^{\star})\}_{i=1}^m$ is an optimal solution satisfying (\ref{equ:KKT3}) by letting $\{\boldsymbol{\lambda}_{i|s}^{\star}=\boldsymbol{\delta}_i^{\star}\}_{i=1}^m$. We utilise Lemma~\ref{lemma:primal_inequality_general} to derive an inequality involving $\{\boldsymbol{x}_i^{r, k}| k=1,\ldots, K  \}_{i=1}^m$ and the above optimal solution:
\begin{sloppypar}
\begin{lemma}
Suppose the estimates $\{\boldsymbol{x}_i^{r, k} \}$ are obtained by performing (\ref{equ:gradient_xi})-(\ref{equ:f_i_approximate}) under the condition that $1/\eta\geq L$. Let $\bar{\boldsymbol{x}}_i^{r, K}=\frac{1}{K}\sum_{k=1}^K \boldsymbol{x}_i^{r, K}$. Then
\begin{align}
 & \sum_{i=1}^m  \frac{1}{K}\sum_{k=0}^{K-1} \hspace{-0.6mm} \frac{1/\eta - \theta \mu}{2}\|\boldsymbol{x}_i^{r,k}-\boldsymbol{x}_i^{\star} \|^2 \hspace{-0.6mm} \nonumber \\
  &+  \sum_{i=1}^m  \frac{1}{4\rho}  \|\rho(\bar{\boldsymbol{x}}_i^{r, K} - \boldsymbol{x}_i^{\star}) + (\boldsymbol{\lambda}_{i|s}^{r+1} -\boldsymbol{\lambda}_{i|s}^{\star} )   \|^2 \nonumber \\
  \hspace{-2mm}&\geq \sum_{i=1}^m \Big[ f_i(\bar{\boldsymbol{x}}_i^{r, K} ) - (\bar{\boldsymbol{x}}_i^{r,K})^T \boldsymbol{\lambda}_{i|s}^{\star} - f_i(\boldsymbol{x}_i^{\star})  \nonumber \\
  & \hspace{9mm} + \frac{1}{K}\sum_{k=0}^{K-1} \Big(\frac{1}{2\eta}  \|\boldsymbol{x}_i^{\star} - \boldsymbol{x}_i^{r, k+1} \|^2  \nonumber\\
\hspace{-3mm}&\hspace{9mm} +\hspace{-0.6mm} \frac{1/\eta - L}{2} \|\boldsymbol{x}_i^{r, k+1}  \hspace{-0.6mm}- \hspace{-0.6mm}\boldsymbol{x}_i^{r, k} \|^2 \hspace{-0.7mm}  + \hspace{-0.7mm}  \frac{1-\theta}{2L}\|  \rho(\hspace{-0.6mm} \boldsymbol{x}_s^{r} \hspace{-0.7mm} - \hspace{-0.6mm}  \boldsymbol{x}_i^{r,k+1} ) \hspace{-0.6mm}\nonumber \\
\hspace{-3mm}&  \hspace{9mm}  -\hspace{-0.6mm}  \boldsymbol{\lambda}_{s|i}^{r} \hspace{-0.6mm} - \boldsymbol{\lambda}_{i|s}^{\star} -\hspace{-0.6mm} (1/\eta)(\boldsymbol{x}_i^{r, k+1} \hspace{-0.6mm} -\hspace{-0.6mm}  \boldsymbol{x}_i^{r, k}) \hspace{-0.6mm}\|^2 \Big) \nonumber \\
 &\hspace{9mm} +  \frac{1}{4\rho}  \|\rho(\bar{\boldsymbol{x}}_i^{r+1, K} - \boldsymbol{x}_i^{\star}) + (\boldsymbol{\lambda}_{i|s}^{r+2} -\boldsymbol{\lambda}_{i|s}^{\star} )   \|^2  \Big],
\label{equ:upper_bound_final}
\end{align}
where $1\geq \theta \geq 0$. 
\label{lemma:twoBounds}
\vspace{0mm}
\end{lemma}
\end{sloppypar}
\begin{proof}
 See Appendix~\ref{appendix:Lemma_upperbound} for the proof. 
\end{proof}

%We note that $\{\theta_i\}$ need to satisfy $\{1>\theta_i>0\}$ to prove the linear convergence rate when the client functions are strongly convex and have Lipschitz continuous gradients. 
Next we show that $\sum_{i=1}^m  \Big[  f_i(\bar{\boldsymbol{x}}_i^{r, K} ) - (\bar{\boldsymbol{x}}_i^{r,K})^T \boldsymbol{\lambda}_{i|s}^{\star} - f_i(\boldsymbol{x}_i^{\star})\Big]$ in (\ref{equ:upper_bound_final}) is lower-bounded by zero in a lemma below:
\begin{lemma}
Suppose $\{\boldsymbol{x}_s^{\star}=\boldsymbol{x}_i^{\star}\}_{i=1}^m$ together with $\{\boldsymbol{\lambda}_{i|s}^{\star} = - \boldsymbol{\lambda}_{s|i}^{\star})\}_{i=1}^m$ is an optimal solution satisfying (\ref{equ:KKT3}) by letting $\{\boldsymbol{\lambda}_{i|s}^{\star}=\boldsymbol{\delta}_i^{\star}\}_{i=1}^m$. 
For any $\{\boldsymbol{x}_i\in \mathbb{R}^d\}_{i=1}^m$, 
\begin{align}
&\sum_{i=1}^m  \Big[   f_i(\boldsymbol{x}_i ) - f_i(\boldsymbol{x}_i^{\star}) - \boldsymbol{x}_i^{T}\boldsymbol{\lambda}_{i|s}^{\star} \Big] \geq 0.
\label{equ:lowerbound}
\end{align}
\label{lemma:lower_bound}
\vspace{-0mm}
\end{lemma}
See Appendix~\ref{appendix:lemma_lowerbound} for the proof. Basically, (\ref{equ:lowerbound}) suggests that the RHS of (\ref{equ:upper_bound_final}) is always lower-bounded by zero.  If needed, the quantity $\sum_{i=1}^m  [   f_i(\bar{\boldsymbol{x}}_i^{r, K} ) - f_i(\boldsymbol{x}_i^{\star}) - (\bar{\boldsymbol{x}}_i^{r, K})^T\boldsymbol{\lambda}_{i|s}^{\star} ]$ can be ignored in (\ref{equ:upper_bound_final}) due to its nonnegativity. 

\begin{figure*}[t!]
\centering
\includegraphics[width=120mm]{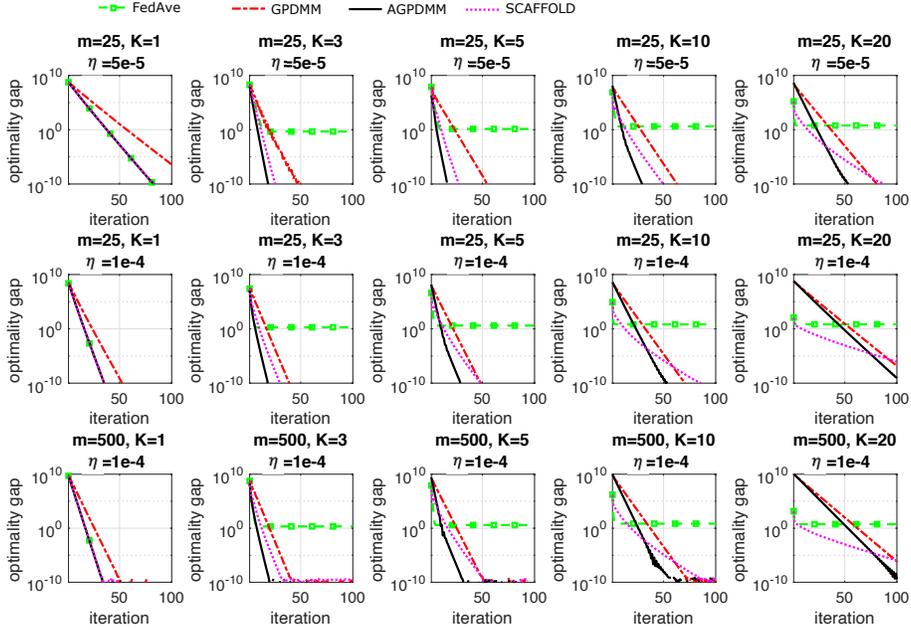}
\psfrag{e}{$\eta$}
\vspace*{-0.2cm}
\caption{\footnotesize{ Performance comparison of FedAve, GPDMM, AGPDMM, and SCAFFOLD for solving a least square problem which is specified by synthetic data.  }}
\label{fig:synthetic}
\vspace*{-0.0cm}
\end{figure*}

\begin{figure*}[t!]
\centering
\includegraphics[width=120mm]{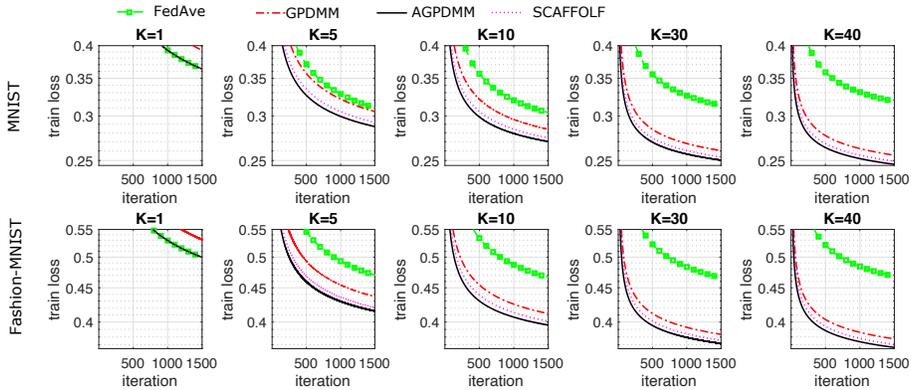}
\vspace*{-0.0cm}
\caption{\footnotesize{Performance comparison for softmax regression over the MNIST and Fashion-MNIST datasets, where the five subplots in the first row are for  MNIST. As classification over Fashion-MNIST is more challenging than that over MNIST, the training losses over Fashion-MNIST are larger than those over MNIST.} }
\label{fig:compare_MNIST}
\vspace*{-0.3cm}
\end{figure*}

\noindent\textbf{Linear convergence results}:  With Lemma~\ref{lemma:twoBounds} and \ref{lemma:lower_bound}, we are ready to show the linear convergence speed for  GPDMM in Alg.~1.  Our main objective is to show that the coefficients before $\|\boldsymbol{x}_i^{r, K} - \boldsymbol{x}_i^{\star}\|^2$ and $ \|\rho(\bar{\boldsymbol{x}}_i^{r+1, K} - \boldsymbol{x}_i^{\star}) + (\boldsymbol{\lambda}_{i|s}^{r+2} -\boldsymbol{\lambda}_{i|s}^{\star} ) \|^2$ on the RHS of  (\ref{equ:upper_bound_final}) are greater than the ones before $\|\boldsymbol{x}_i^{r-1, K} - \boldsymbol{x}_i^{\star}\|^2$ and $ \|\rho(\bar{\boldsymbol{x}}_i^{r, K} - \boldsymbol{x}_i^{\star}) + (\boldsymbol{\lambda}_{i|s}^{r+1} -\boldsymbol{\lambda}_{i|s}^{\star} ) \|^2$ on the LHS of (\ref{equ:upper_bound_final}) for each client $i$. The other quantities in  (\ref{equ:upper_bound_final}) are either dropped or combined to produce the above mentioned ones. We  summarise the results in a theorem below:  
\begin{theorem} Suppose the estimates $\{\boldsymbol{x}_i^{r, k} \}$ are obtained by performing (\ref{equ:gradient_xi})-(\ref{equ:f_i_approximate}) under the condition that $1/\eta > L\geq \mu>0$.  Let $Q^r$, $r\geq 1$,  be 
\begin{align}
&Q^r = \sum_{i=1}^m \Big[\frac{1/\eta - \theta\mu}{2K}  \| \boldsymbol{x}_i^{r-1, K} - \boldsymbol{x}_i^{\star} \|^2 \nonumber \\
	&\hspace{2mm}+ \hspace{-0.7mm} (\frac{1}{4\rho} \hspace{-0.7mm}-\hspace{-0.7mm} \frac{\gamma_{2}}{2})  \|\rho(\bar{\boldsymbol{x}}_i^{r, K} \hspace{-0.7mm}-\hspace{-0.7mm} \boldsymbol{x}_i^{\star}) \hspace{-0.7mm}+\hspace{-0.7mm} (\boldsymbol{\lambda}_{i|s}^{r+1} \hspace{-0.7mm}-\hspace{-0.7mm} \boldsymbol{\lambda}_{i|s}^{\star} )   \|^2 \Big], \label{equ:} 
\end{align}
where 
\begin{align}
& \gamma_{2} \hspace{-0.7mm}=\hspace{-0.7mm} \min\left( \frac{\theta\mu \phi}{2\rho^2},  \frac{\gamma_1\eta^2}{2}  \right), \label{equ:gamma_2} 
\end{align}
where $1 \hspace{-0.7mm}>\hspace{-0.7mm} \theta \hspace{-0.7mm}> \hspace{-0.7mm}0$, $ 1 \hspace{-0.7mm}>\hspace{-0.7mm}\phi \hspace{-0.7mm}> \hspace{-0.7mm}0$ satisfy $\frac{\theta\mu \phi}{4\rho^2} < \frac{1}{4\rho}$, and  
\begin{align}
\gamma_{1} \hspace{-0.7mm}=\hspace{-0.7mm} \min\left(\frac{1-\theta}{2L \eta^2} , \frac{1/\eta \hspace{-0.7mm}-\hspace{-0.7mm} L}{2}\right). \label{equ:gamma_1}
\end{align}
Then
\begin{align}
Q^{k+1} \leq \beta Q^k, \label{equ:linear_conv}
\end{align}
where $0<\beta<1$ is computed as 
\begin{align}
\beta &= \max\left( \frac{1/(4\rho)- \gamma_{2}/2 }{1/(4\rho)},  \frac{ 1/\eta - \theta\mu }{ 1/\eta - \theta\mu\phi}  \right).
\nonumber
\end{align}
 %The parameters $\{\theta_i\}_{i=1}^m$ and $\{\phi_i\}_{i=1}^m$ are the free parameters to be specified. 
\label{theorem:linear_conv}
\end{theorem}

\begin{proof}
See Appendix~\ref{appendix:linear_conv} for the proof. The constraint $0<\beta<1$ is guaranteed by the fact that $ 1/\eta> L\geq\mu > \theta\mu$,  
$\frac{1}{4\rho} > \frac{\theta\mu \phi}{4\rho^2} \geq \frac{\gamma_{2}}{2} $, and $1>\phi>0$.
\end{proof}

\noindent \textbf{Sublinear convergence results}: For the special case that the client functions are not strongly convex, (i.e., $\mu=0$ in (\ref{equ:muStrong})), the method exhibits sublinear convergence speed.  The convergence rate can be  characterised by setting  $\mu=0$ and $\theta=0$ in (\ref{equ:upper_bound_final}), performing summation from $r=1$ to $r=R$, and applying Jensen's inequality. We summarise the results in a theorem below: 
\begin{theorem} Consider the special case $\mu=0$ in (\ref{equ:muStrong}) for all clients.  Suppose the estimates $\{\boldsymbol{x}_i^{r,k}\}$ are obtained by performing (\ref{equ:gradient_xi})-(\ref{equ:f_i_approximate}) under the condition that $1/\eta > L$.  Let $\bar{\boldsymbol{x}}_i^{R, K} =\frac{1}{R}\sum_{r=1}^R\bar{\boldsymbol{x}}_i^{r, K}=\frac{1}{RK}\sum_{r=1}^{R}\sum_{k=1}^K \boldsymbol{x}_i^{r,k}$ and $\bar{\boldsymbol{\lambda}}_{i|s}^{R} =\frac{1}{R}\sum_{r=1}^R \boldsymbol{\lambda}_{i|s}^{r+1}$. Then
\begin{align}
\hspace{-3mm} &\lim_{R\rightarrow \infty}\sum_{i=1}^m \Big[ f_i(\bar{\boldsymbol{x}}_i^{R,K} ) \hspace{-0.7mm}-\hspace{-0.7mm} \boldsymbol{\lambda}_{i|s}^{\star, T}\bar{\boldsymbol{x}}_i^{R,K}   \hspace{-0.7mm}-\hspace{-0.7mm} f_i(\boldsymbol{x}_i^{\star}) \Big] \hspace{-0.7mm}=\hspace{-0.7mm} \mathcal{O}(1/R) \label{equ:sublinear1} 
\end{align}
\begin{align}
\hspace{-3mm} &\lim_{R\rightarrow \infty}\sum_{i=1}^m \Big[  \hspace{-0.6mm}\frac{\gamma_{1}\eta^2}{2} \|\bar{\boldsymbol{\lambda}}_{i|s}^{R} \hspace{-0.7mm}-\hspace{-0.7mm} \boldsymbol{\lambda}_{i|s}^{\star}  \|^2 \Big] = \mathcal{O}(1/R),  \label{equ:sublinear2}
\end{align}
where $\gamma_{1}$ is given by (\ref{equ:gamma_1}) by setting $\theta=0$. 
\label{theorem:sublinear}
\end{theorem}
\begin{proof}
See Appendix~\ref{appendix:sublinear} for the proof. 
\end{proof}

It is clear from Lemma~\ref{lemma:lower_bound} that the LHS of (\ref{equ:sublinear1}) is lower-bounded by zero for all $R\geq 1$. When $R$ approaches to infinity, we have $\{\nabla f_i(\bar{\boldsymbol{x}}_i^{R,K} ) = \boldsymbol{\lambda}_{i|s}^{\star}\}_{i=1}^m$, showing that the limiting point $\{\bar{\boldsymbol{x}}_i^{R, K}\}$ is in fact the optimal solution.   

\vspace{-2mm}
\section{Experimental Results}
\vspace{-2mm}
Two experiments were conducted to evaluate FedAve \cite{McMahan17}, GPDMM, AGPDMM, and SCAFFOLD.  Inexact FedSplit is not considered because of its poor performance as demonstrated in Fig.~\ref{fig:FedSplit}. The two experiments are least square minimisation over synthetic data and softmax regression over MNIST and Fashion-MNIST datasets, respectively.  

\vspace{-2mm}
\subsection{Least square minimisation over synthetic data}
\vspace{-2mm}
\label{subsec:least_square}
We consider solving a least square problem over a centralised network (see  \cite{Pathak2021} for a similar experimental setup).  The objective function $f_i(\boldsymbol{x}_i)$ takes the form $f_i(\boldsymbol{x}_i) = \frac{1}{2}\|\boldsymbol{A}_i \boldsymbol{x}_i - \boldsymbol{b}_i \|^2$, where $\boldsymbol{A}_i\in \mathbb{R}^{5000\times 500}$ are generated element-wise from a Normal distribution. %As $5000\gg 500$, the function $f_i$  is guaranteed by non-asymptotic random matrix theory \cite{Wainwright19book}  to be strongly convex and have Lipschitz continuous gradient. 
The vector $\boldsymbol{b}_i $ is obtained by letting $\boldsymbol{b}_i = \boldsymbol{A}_i\boldsymbol{y}_0+\boldsymbol{v}_i$, where $\boldsymbol{y}_0$ is a predefined vector and $\boldsymbol{v}_i\sim N(0, 0.25\boldsymbol{I}_{5000\times 5000})$.

In all four methods,  $\{\boldsymbol{x}_i\}$ and $\boldsymbol{x}_s$ were initialised to be zero.  In addition, the other hyper-parameters $\eta=\{5e-5, 1e-4\}$, $m=\{25, 500\}$, and $K=\{1,3,5,10, 20\}$ were tested. The parameter $\eta_g$ in SCAFFOLD was set to $\eta_g=1$ to be in line with the setup $\rho=\frac{1}{\eta}$ of AGPDMM in (\ref{equ:xs_AGPDMM_K1}). Finally, the control covariates of SCAFFOLD were initialised to be zero.

Fig.~\ref{fig:synthetic} displays the convergence results for the four methods. Firstly, one observes that FedAve has poor performance when $K>1$, which is due to the functional heterogeneity across the clients nodes (i.e., the global optimal solution $\boldsymbol{x}_s^{\ast}$ is inconsistent with the optimal solutions of individual client functions \cite{Pathak2021}). Secondly, it is clear that AGPDMM converges faster than GPDMM for all tested $K$ values. As explained in Section~\ref{sec:GPDMM}, the performance gain of AGPDMM is due to the fact that at each iteration $r$, the global estimate $\boldsymbol{x}_s^r$ instead of the individual estimate $\boldsymbol{x}_i^{r-1, K}$ is utilised to perform approximate optimisations at the client $i$. Thirdly, one can also find from the figure that AGPDMM converges faster than SCAFFOLD when $K>1$. This might be because the computation of $\boldsymbol{\lambda}_{s|i}^{r+1}$ in AGPDMM utilises both $\{\boldsymbol{x}_s^{r} - \boldsymbol{x}_i^{r, K} \}$ and $\{\boldsymbol{x}_s^{r+1} - \boldsymbol{x}_i^{r, K} \}$  while the computation of $\boldsymbol{c}^{r+1}$ in SCAFFOLD utilises only $\{\boldsymbol{x}_s^{r} - \boldsymbol{x}_i^{r, K} \}$. When $K=1$, both methods have identical performance as FedAve. This is because both methods have the identical update expression for the estimate $\boldsymbol{x}_s^{r+1}$, which is in fact the expression of vanilla gradient descent in FedAve.    

\vspace{-2mm}
\subsection{Softmax regression over MNIST and Fashion-MNIST }
\vspace{-1mm}
In this experiment, we consider performing softmax regression (i.e., a convex optimisation problem) over the MNIST and Fashion-MNIST datasets,  where each dataset has 10 classes. The number of clients is set to be $m=10$ for each dataset, where each client carries the training images of a single class. The above setup implies that the distributions of the training data are heterogeneous across the different clients. 

Similarly to the first experiment,  $\{\boldsymbol{x}_i\}$ and $\boldsymbol{x}_s$ were initialised to be zero in the four methods . The other hyper-parameters $\eta=0.05$ and $K=\{1, 5, 10, 30, 40\}$ were tested. The parameter $\eta_g$ and the control covariates for  SCAFFOLD were set as in the first experiment. At each gradient step of an iteration at a client node, a mini-batch of 300 training samples was utilised to compute the gradient and update the model parameters accordingly. It is noted that the mini-batches were taken in a pre-defined order instead of in a random manner to remove any effect of randomness. That is, the training procedure is deterministic.  

The training results and validation accuracies are summarised in Fig.~\ref{fig:compare_MNIST} and Table~\ref{tab:val_acc}, respectively.  One observes that for each dataset, the training loss of each method improves gradually as $K$ increases from 1 to $40$ except FedAve. In addition, it is clear that AGPDMM performs the best w.r.t. the training loss.  As for validation accuracy, AGPGMM outperforms others for most scenarios except $K=10$ for Fashion-MNIST. SCAFFOLD performs slightly better than GPDMM. The above phenomenon suggests that the initialisation for each iteration at the client side is crucial for Inexact PDMM.

\begin{table}[t]
\caption{\small Validation accuracy (in percentage) of the three methods for the MNIST and Fashion-MNIST datasets} 
\label{tab:val_acc}
\centering
\begin{tabular}{|c|c|c|c|c|c|c|}
\hline
& \hspace{-1.5mm}\scriptsize{ K } \hspace{-1.5mm} & \footnotesize{1} & \footnotesize{5}& \footnotesize{10}& \footnotesize{30} & \footnotesize{40} 
 \\ 
\hline
{\scriptsize \multirow{4}{*}{\rotatebox{90}{MNIST}}}    & \hspace{-1.5mm}\scriptsize{ FedAve } \hspace{-1.5mm} &  \footnotesize{90.80} & \footnotesize{91.70} & \footnotesize{91.67}& \footnotesize{91.32} & \footnotesize{91.16}  \\ % Entering row contents Midnight&7&-3& 5&3&-1&-3&5\\
\cline{2-7}
& \hspace{-1.5mm}\scriptsize{ GPDMM } \hspace{-1.5mm} & \footnotesize{90.25}  & \footnotesize{91.92} & \footnotesize{92.20}& \footnotesize{92.46} & \footnotesize{92.52}  \\ % Entering row contents Midnight&7&-3& 5&3&-1&-3&5\\
\cline{2-7}
& \hspace{-1.5mm}\scriptsize{SCAFFOLD}  \hspace{-1.5mm} & \footnotesize{90.80}  & \footnotesize{92.10} & \footnotesize{92.29} & \footnotesize{92.53}  & \footnotesize{92.59}  \\ 
\cline{2-7}
& \hspace{-1.5mm}\scriptsize{AGPDMM }  \hspace{-1.5mm} & \footnotesize{90.80}  & \footnotesize{\textbf{92.14}} & \footnotesize{\textbf{92.37}} & \footnotesize{\textbf{92.61}} & \footnotesize{\textbf{92.64}}  \\ 
\hline 
\hline
{\scriptsize \multirow{4}{*}{\rotatebox{90}{Fashion-MNIST}}} 
 & \hspace{-1.5mm}\scriptsize{ FedAve } \hspace{-1.5mm} &  \footnotesize{82.24} & \footnotesize{83.08} & \footnotesize{83.13}& \footnotesize{83.09} & \footnotesize{82.83}  \\ % Entering row contents Midnight&7&-3& 5&3&-1&-3&5\\
\cline{2-7}
& \hspace{-1.5mm}\scriptsize{ GPDMM } \hspace{-1.5mm} &  \footnotesize{81.43}  & \footnotesize{83.64} & \footnotesize{84.18}& \footnotesize{84.58} & \footnotesize{84.64}  \\ % Entering row contents Midnight&7&-3& 5&3&-1&-3&5\\
\cline{2-7}
& \hspace{-1.5mm}\scriptsize{SCAFFOLD}  \hspace{-1.5mm} &  \footnotesize{82.24}  & \footnotesize{83.97} & \footnotesize{\textbf{84.49}} & \footnotesize{84.66}  & \footnotesize{84.65}  \\ 
\cline{2-7}
& \hspace{-1.5mm}\scriptsize{AGPDMM }  \hspace{-1.5mm} &  \footnotesize{82.24}  & \footnotesize{\textbf{84.08}} & \footnotesize{{84.46}} & \footnotesize{\textbf{84.67}} & \footnotesize{\textbf{84.65}}  \\ 
\hline 
\end{tabular}
\vspace{-0mm}
\end{table}

%compare_simulation

\section{Conclusions}
In this paper, we first showed that PDMM reduces to FedSplit when applied to a centralised network. The poor reported performance of Inexact FedSplit in \cite{Pathak2021} is analysed, which was found to be due to the improper parameter initialisation at the client side. Two versions of Inexact PDMM were then proposed to correct the convergence issue of Inexact FedSplit, which are GPDMM and AGPDMM. The main difference between the methods is that at each iteration $r$, AGPDMM utilises the global estimate $\boldsymbol{x}_s^r$ to conduct approximate optimisations at the client slide, which is more informative than the individual estimates $\{\boldsymbol{x}_i^{r-1, K}\}$. Linear and sublinear convergence rates are established for GPDMM for any number ($K> 0$) of approximate optimisations conducted at the client side per iteration. It is also shown analytically that when $K=1$, both AGPDMM and SCAFFOLD reduce to the vanilla gradient descent operation under proper parameter selection. Therefore, convergence results of the classical vanilla gradient descent operation apply directly to AGPDMM when $K=1$. Experimental results show that AGPDMM converges faster than both SCAFFOLD and GPDMM.

One future work would be to provide a convergence analysis for AGPDMM when $K>1$. One can also extend the deterministic  analysis for GPDMM to the stochastic scenario. 

\appendices

\vspace{-3mm}
\section{Proof for Lemma~\ref{lemma:primal_inequality_general}}
\label{appendix:lemma_ineq}

Before presenting the proof, we first introduce two lemmas that will be needed later on:
\begin{lemma} For any $\boldsymbol{y}_i \in\mathbb{R}^d$, $i=1,\ldots, 4$, the following equality holds 
\begin{align}
&(\boldsymbol{y}_1-\boldsymbol{y}_2)^T(\boldsymbol{y}_3-  \boldsymbol{y}_4)  \nonumber \\
&= \frac{1}{2}\left(\|\boldsymbol{y}_1 \hspace{-0.7mm}+\hspace{-0.7mm}\boldsymbol{y}_3 \|^2 \hspace{-0.7mm}-\hspace{-0.7mm} \|\boldsymbol{y}_2 \hspace{-0.7mm}+\hspace{-0.7mm} \boldsymbol{y}_4 \|^2 \hspace{-0.7mm}-\hspace{-0.7mm} \|\boldsymbol{y}_2\hspace{-0.7mm}+\hspace{-0.7mm}\boldsymbol{y}_3 \|^2 \hspace{-0.7mm}+\hspace{-0.7mm}  \|\boldsymbol{y}_2\hspace{-0.7mm}+\hspace{-0.7mm}\boldsymbol{y}_4 \|^2\right). \nonumber
\end{align}
\label{lemma:identity}
\end{lemma}

\begin{lemma}
Suppose $f_i$ has the Lipschitz continuous gradient $L>0$. Then the following inequality
\begin{align}
&\hspace{-3mm}f_i(\boldsymbol{y}_i) \hspace{-0.5mm}\leq\hspace{-0.5mm} f_i(\boldsymbol{x}_i) \hspace{-0.5mm}+\hspace{-0.5mm} \nabla f_i(\boldsymbol{x}_i)^T(\boldsymbol{y}_i\hspace{-0.5mm}-\hspace{-0.5mm}\boldsymbol{x}_i) \hspace{-0.5mm}+\hspace{-0.5mm} \frac{L}{2}\|\boldsymbol{x}_i \hspace{-0.5mm}-\hspace{-0.5mm} \boldsymbol{y}_i \|^2 \nonumber 
\end{align}
holds, which is a consequence of the inequality (\ref{equ:gradLips2})  (see \cite{Zhou18Duality}). 
\label{lemma:gradLips}
\end{lemma}
\vspace{-5mm}
\begin{proof}
We now describe the proof for Lemma~\ref{lemma:primal_inequality_general}.  The expression $f_i(\boldsymbol{x}_i) -  f_i(\boldsymbol{x}_i^{r, k+1})$ for client $i$ can be lower-bounded to be
 %From (\ref{equ:f_i_approximate}) and Lemma~\ref{lemma:function_uppperBound}, we have 
  \begin{align}
& \hspace{-1mm} f_i(\boldsymbol{x}_i) -  f_i(\boldsymbol{x}_i^{r, k+1}) \nonumber \\
 \stackrel{(a)}{\geq}& \Big[f_i(\boldsymbol{x}_i^{r, k}) + (\boldsymbol{x}_i - \boldsymbol{x}_i^{r,k})^T \nabla f_i (\boldsymbol{x}_i^{r,k}) + \frac{\theta \mu }{2}\|\boldsymbol{x}_i^{r,k}-\boldsymbol{x}_i \|^2 \hspace{2mm} 
 \nonumber \\
&+ \frac{1-\theta}{2L}\| \nabla f_i(\boldsymbol{x}_i^{r,k})- \nabla f_i ( \boldsymbol{x}_i) \|^2  \Big] - \hspace{-0.7mm}  \Big[ f_i(\boldsymbol{x}_i^{r,k}) \hspace{-0.7mm}  \nonumber \\
& +  \hspace{-0.7mm} (\boldsymbol{x}_i^{r,k+1} \hspace{-0.7mm}-\hspace{-0.7mm} \boldsymbol{x}_i^{r,k})^T \nabla f_i(\boldsymbol{x}_i^{r,k})  \hspace{-0.7mm}+ \hspace{-0.7mm} \frac{L}{2} \|\boldsymbol{x}_i^{r,k+1}  \hspace{-0.7mm}- \hspace{-0.7mm}\boldsymbol{x}_i^{r,k} \|^2 \Big] \nonumber \\ 
=& \hspace{-0.2mm} (\boldsymbol{x}_i  \hspace{-0.7mm}-\hspace{-0.7mm} \boldsymbol{x}_i^{r, k+1})^T \nabla \hspace{-0.2mm} f_i(\boldsymbol{x}_i^{r,k}) \hspace{-0.7mm}  + \hspace{-0.7mm} \frac{\theta\mu}{2}\|\boldsymbol{x}_i^{r,k} \hspace{-0.7mm}-\hspace{-0.7mm}  \boldsymbol{x}_i \|^2 \hspace{-0.8mm} 
\nonumber \\
&- \hspace{-0.8mm} \frac{L}{2} \|\boldsymbol{x}_i^{r,k+1}  \hspace{-0.8mm}  - \hspace{-0.7mm}\boldsymbol{x}_i^{r, k} \|^2 \hspace{-0.7mm}+\hspace{-0.7mm} \frac{1-\theta}{2L}\| \nabla f_i(\boldsymbol{x}_i^{r,k})- \nabla f_i ( \boldsymbol{x}_i) \|^2  \nonumber 
\end{align}
\begin{align}
\stackrel{(b)}{=}&  \hspace{-0.6mm} (\boldsymbol{x}_i \hspace{-0.7mm}-\hspace{-0.7mm} \boldsymbol{x}_i^{r,k+1})^T \Big(  \rho(\hspace{-0.6mm} \boldsymbol{x}_s^{r} \hspace{-0.7mm} - \hspace{-0.6mm}  \boldsymbol{x}_i^{r,k+1} ) \hspace{-0.6mm} -\hspace{-0.6mm}  \boldsymbol{\lambda}_{s|i}^{r} \hspace{-0.7mm} -\hspace{-0.7mm}  \frac{1}{\eta} (\boldsymbol{x}_i^{r, k+1} \hspace{-0.7mm} -\hspace{-0.7mm}  \boldsymbol{x}_i^{r, k})\Big) \nonumber\\
&+ \hspace{-0.6mm} \frac{\theta\mu}{2}\|\boldsymbol{x}_i^{r, k}-\boldsymbol{x}_i \|^2   \hspace{-0.6mm} -\hspace{-0.6mm} \frac{L}{2} \|\boldsymbol{x}_i^{r, k+1}  \hspace{-0.6mm}- \hspace{-0.6mm}\boldsymbol{x}_i^{r, k} \|^2 \nonumber \\
& + \frac{1-\theta}{2L}\| \nabla f_i(\boldsymbol{x}_i^{r,k})- \nabla f_i ( \boldsymbol{x}_i) \|^2  \nonumber \\
\stackrel{(c)}{=}&  \hspace{-0.6mm} (\boldsymbol{x}_i \hspace{-0.6mm}-\hspace{-0.6mm} \boldsymbol{x}_i^{r, k+1})^T  [ \rho(\hspace{-0.6mm} \boldsymbol{x}_s^{r} \hspace{-0.7mm} - \hspace{-0.6mm}  \boldsymbol{x}_i^{r,k+1} ) \hspace{-0.6mm} -\hspace{-0.6mm}  \boldsymbol{\lambda}_{s|i}^{r}] \hspace{-0.6mm}+\hspace{-0.6mm} \frac{1}{2\eta}  
\|\boldsymbol{x}_i - \boldsymbol{x}_i^{r, k+1} \|^2 \nonumber\\
& - \hspace{-0.6mm} \frac{1/\eta - \theta\mu}{2}\|\boldsymbol{x}_i^{r,k}-\boldsymbol{x}_i \|^2 \hspace{-0.6mm} +\hspace{-0.6mm} \frac{1/\eta- L}{2} \|\boldsymbol{x}_i^{r, k+1}  \hspace{-0.6mm}- \hspace{-0.6mm}\boldsymbol{x}_i^{r, k} \|^2 \nonumber \\
&+ \frac{1-\theta}{2L}\| \nabla f_i(\boldsymbol{x}_i^{r,k})- \nabla f_i ( \boldsymbol{x}_i) \|^2
 \end{align}
 where step $(a)$ follows from (\ref{equ:gradLips2})- (\ref{equ:muStrong}) and Lemma~\ref{lemma:gradLips}, which are due to the fact that $f_i$ is $\mu$-convex ($\mu\geq 0$) and has   Lipschitz continuous gradient $L\geq \mu$.  The parameter  $\theta$ satisfy $\{1\geq \theta\geq 0\}$. Step $(b)$ uses the optimality condition (\ref{equ:opti_r}). Step $(c)$ makes use of Lemma~\ref{lemma:identity}. 
 The proof is complete.   
 \end{proof}
  
 \section{Proof for Lemma~\ref{lemma:twoBounds}}
 \label{appendix:Lemma_upperbound}

\begin{proof}
 Invoking Lemma~\ref{lemma:primal_inequality_general} with $\boldsymbol{x}_i=\boldsymbol{x}_i^{\star}$, summing over all the clients and all gradient steps $i=1,\ldots, K$,  for the iteration $r$, and rearranging the quantities, we obtain   
  \begin{align}
 & \sum_{i=1}^m  \frac{1}{K}\sum_{k=0}^{K-1} \hspace{-0.6mm} \frac{1/\eta - \theta \mu}{2}\|\boldsymbol{x}_i^{r,k}-\boldsymbol{x}_i^{\star} \|^2 \hspace{-0.6mm} \nonumber \\
 \hspace{-2mm}&\geq \sum_{i=1}^m \hspace{-0.5mm} \frac{1}{K}\hspace{-0.5mm} \sum_{k=0}^{K-1} \Big[ f_i(\boldsymbol{x}_i^{r, k+1} ) - f_i(\boldsymbol{x}_i^{\star}) \hspace{-0.6mm} + \frac{1}{2\eta}  
\|\boldsymbol{x}_i^{\star} - \boldsymbol{x}_i^{r, k+1} \|^2  \nonumber \\
& \hspace{10mm} - (  \boldsymbol{x}_i^{r, k+1} - \boldsymbol{x}_i^{\star} )^T  [ \rho(\hspace{-0.6mm} \boldsymbol{x}_s^{r} \hspace{-0.7mm} - \hspace{-0.6mm}  \boldsymbol{x}_i^{r,k+1} ) \hspace{-0.6mm} -\hspace{-0.6mm}  \boldsymbol{\lambda}_{s|i}^{r}] \nonumber \\
&+\hspace{-0.6mm} \frac{1/\eta \hspace{-0.5mm}-  \hspace{-0.5mm} L}{2} \|\boldsymbol{x}_i^{r, k+1}  \hspace{-0.7mm}- \hspace{-0.7mm}\boldsymbol{x}_i^{r, k} \|^2 \hspace{-0.7mm} + \hspace{-0.7mm}  \frac{1  \hspace{-0.5mm}- \hspace{-0.5mm} \theta}{2L}\| \nabla f_i(\boldsymbol{x}_i^{r, k})  \hspace{-0.7mm} -\hspace{-0.7mm}  \nabla f_i ( \boldsymbol{x}_i^{\star}) \|^2 \Big] \nonumber\\
 \hspace{-2mm}&\stackrel{(a)}{=} \sum_{i=1}^m \hspace{-0.5mm} \frac{1}{K}\hspace{-0.5mm} \sum_{k=0}^{K-1} \Big[ f_i(\boldsymbol{x}_i^{r, k+1} ) - f_i(\boldsymbol{x}_i^{\star}) \hspace{-0.6mm} + \frac{1}{2\eta}  
\|\boldsymbol{x}_i^{\star} - \boldsymbol{x}_i^{r, k+1} \|^2  \nonumber\\
& \hspace{10mm} - (  \boldsymbol{x}_i^{r, k+1} - \boldsymbol{x}_i^{\star} )^T  [ \rho(\hspace{-0.6mm} \boldsymbol{x}_s^{r} \hspace{-0.7mm} - \hspace{-0.6mm}  \boldsymbol{x}_i^{r,k+1} ) \hspace{-0.6mm} -\hspace{-0.6mm}  \boldsymbol{\lambda}_{s|i}^{r}] \nonumber \\
\hspace{-3mm}&\hspace{9mm} +\hspace{-0.6mm} \frac{1/\eta - L}{2} \|\boldsymbol{x}_i^{r, k+1}  \hspace{-0.6mm}- \hspace{-0.6mm}\boldsymbol{x}_i^{r, k} \|^2 \hspace{-0.7mm}  + \hspace{-0.7mm}  \frac{1-\theta}{2L}\|  \rho(\hspace{-0.6mm} \boldsymbol{x}_s^{r} \hspace{-0.7mm} - \hspace{-0.6mm}  \boldsymbol{x}_i^{r,k+1} ) \hspace{-0.6mm}\nonumber \\
\hspace{-3mm}&  \hspace{9mm}  -\hspace{-0.6mm}  \boldsymbol{\lambda}_{s|i}^{r} \hspace{-0.6mm} - \boldsymbol{\lambda}_{i|s}^{\star} -\hspace{-0.6mm} (1/\eta)(\boldsymbol{x}_i^{r, k+1} \hspace{-0.6mm} -\hspace{-0.6mm}  \boldsymbol{x}_i^{r, k}) \hspace{-0.6mm}\|^2 \Big] \nonumber \\
 \hspace{-2mm}&\stackrel{(b)}{=} \sum_{i=1}^m \Big[ f_i(\bar{\boldsymbol{x}}_i^{r, K} ) - f_i(\boldsymbol{x}_i^{\star}) \hspace{-0.6mm} + \frac{1}{K}\sum_{k=0}^{K-1} \frac{1}{2\eta}  
\|\boldsymbol{x}_i^{\star} - \boldsymbol{x}_i^{r, k+1} \|^2  \nonumber\\
& \hspace{10mm} - (  \bar{\boldsymbol{x}}_i^{r, K} - \boldsymbol{x}_i^{\star} )^T \boldsymbol{\lambda}_{i|s}^{r+1} \nonumber \\
\hspace{-3mm}&\hspace{9mm} +\hspace{-0.6mm} \frac{1/\eta - L}{2} \|\boldsymbol{x}_i^{r, k+1}  \hspace{-0.6mm}- \hspace{-0.6mm}\boldsymbol{x}_i^{r, k} \|^2 \hspace{-0.7mm}  + \hspace{-0.7mm}  \frac{1-\theta}{2L}\|  \rho(\hspace{-0.6mm} \boldsymbol{x}_s^{r} \hspace{-0.7mm} - \hspace{-0.6mm}  \boldsymbol{x}_i^{r,k+1} ) \hspace{-0.6mm}\nonumber \\
\hspace{-3mm}&  \hspace{9mm}  -\hspace{-0.6mm}  \boldsymbol{\lambda}_{s|i}^{r} \hspace{-0.6mm} - \boldsymbol{\lambda}_{i|s}^{\star} -\hspace{-0.6mm} (1/\eta)(\boldsymbol{x}_i^{r, k+1} \hspace{-0.6mm} -\hspace{-0.6mm}  \boldsymbol{x}_i^{r, k}) \hspace{-0.6mm}\|^2 \Big] ,
\label{equ:upper_bound1}
 \end{align}
 where step $(a)$ uses the optimality condition (\ref{equ:opti_r}) and  $ \{\nabla f_i ( \boldsymbol{x}_i^{\star}) = \boldsymbol{\lambda}_{i|s}^{\star}\}_{i=1}^m$. Step $(b)$ is obtained by employing Jensen's inequality, $ \bar{\boldsymbol{x}}_i^{r,K} = \frac{1}{K}\sum_{k=1}^K \boldsymbol{x}_i^{r,k}$, and $ \boldsymbol{\lambda}_{i|s}^{r+1}= \rho(\hspace{-0.6mm} \boldsymbol{x}_s^{r} \hspace{-0.7mm} - \hspace{-0.6mm}  \bar{\boldsymbol{x}}_i^{r,K} ) \hspace{-0.6mm} -\hspace{-0.6mm}  \boldsymbol{\lambda}_{s|i}^{r}$.

To further simplify (\ref{equ:upper_bound1}), we first present a lemma below:
\begin{lemma}
Suppose the estimates $\{\boldsymbol{x}_i^{r, k}\}_{k=1}^{K}$ are obtained by performing (\ref{equ:gradient_xi})-(\ref{equ:f_i_approximate}) under the condition that $1/\eta \geq L$.  Then the expression $\sum_{i=1}^m ( \bar{\boldsymbol{x}}_i^{r, K} - \boldsymbol{x}_i^{\star} )^T  \boldsymbol{\lambda}_{i|s}^{r+1} $ in the RHS of (\ref{equ:upper_bound1}) can be alternatively represented
as
\begin{align}
 &\hspace{-0mm} 2\sum_{i=1}^m ( \bar{\boldsymbol{x}}_i^{r,K} - \boldsymbol{x}_i^{\star} )^T  \boldsymbol{\lambda}_{i|s}^{r+1} \nonumber \\
&= 2\sum_{i=1}^m \boldsymbol{\lambda}_{i|s}^{\star} \bar{\boldsymbol{x}}_i^{r, K} + \sum_{i=1}^m  \frac{1}{2\rho}  \|\rho(\bar{\boldsymbol{x}}_i^{r, K} - \boldsymbol{x}_i^{\star}) + (\boldsymbol{\lambda}_{i|s}^{r+1} -\boldsymbol{\lambda}_{i|s}^{\star} )   \|^2  \nonumber \\
& \hspace{3mm}  -   \sum_{i=1}^m \frac{1}{2\rho}  \|\rho(\bar{\boldsymbol{x}}_i^{r+1, K} - \boldsymbol{x}_i^{\star}) + (\boldsymbol{\lambda}_{i|s}^{r+2} -\boldsymbol{\lambda}_{i|s}^{\star} )   \|^2.
\label{equ:client_ineq_final}
\end{align}
We postpone the proof for Lemma~\ref{lemma:client_ineq_final} in Appendix~\ref{appendix:proof_client_ineq}. 
\label{lemma:client_ineq_final}
\end{lemma}

Plugging (\ref{equ:client_ineq_final}) into (\ref{equ:upper_bound1}) and rearranging the quantities produces (\ref{equ:upper_bound_final}). The proof is complete. \end{proof}

\vspace{-3mm}
\section{Proof for Lemma~\ref{lemma:client_ineq_final}}
\label{appendix:proof_client_ineq}
\vspace{-3mm}
\begin{proof}
In the first step, we derive two different but mathematically equivalent expressions for the quantity $\sum_{i=1}^m ( \bar{\boldsymbol{x}}_i^{r,K}  - \boldsymbol{x}_i^{\star} )^T  \boldsymbol{\lambda}_{i|s}^{r+1}$.
Firstly,  by plugging the expressions $\{\boldsymbol{\lambda}_{i|s}^{r+1} = \rho (\boldsymbol{x}_s^r \hspace{-0.7mm}-\hspace{-0.7mm} \bar{\boldsymbol{x}}_i^{r,K} ) -\boldsymbol{\lambda}_{s|i}^r \}$ into $\sum_{i=1}^m (  \bar{\boldsymbol{x}}_i^{r,K} - \boldsymbol{x}_i^{\star} )^T  \boldsymbol{\lambda}_{i|s}^{r+1}$, we have 
\begin{align}
&\sum_{i=1}^m (  \bar{\boldsymbol{x}}_i^{r,K} - \boldsymbol{x}_i^{\star} )^T  \boldsymbol{\lambda}_{i|s}^{r+1} \nonumber \\
 &=\sum_{i=1}^m ( \rho ( \boldsymbol{x}_s^r - \bar{\boldsymbol{x}}_i^{r,K} ) 
-\boldsymbol{\lambda}_{s|i}^r)^T (\bar{\boldsymbol{x}}_i^{r,K} - \boldsymbol{x}_i^{\star})   \nonumber \\
&= \sum_{i=1}^m \Big( \rho  ( \boldsymbol{x}_s^r - \bar{\boldsymbol{x}}_i^{r, K} ) 
+\boldsymbol{\lambda}_{s|i}^{r+1}-\boldsymbol{\lambda}_{s|i}^r\Big)^T  (\bar{\boldsymbol{x}}_i^{r, K} - \boldsymbol{x}_i^{\star})  \hspace{-0.6mm} \nonumber \\
&\hspace{3mm}-\hspace{-0.6mm} \sum_{i=1}^m \boldsymbol{\lambda}_{s|i}^{r+1} (\bar{\boldsymbol{x}}_i^{r, K} - \boldsymbol{x}_i^{\star}) \nonumber \\
&= \sum_{i=1}^m \Big( \rho ( \boldsymbol{x}_s^r - \boldsymbol{x}_s^{r+1} ) 
+\boldsymbol{\lambda}_{s|i}^{r+1}-\boldsymbol{\lambda}_{s|i}^r\Big)^T  (\bar{\boldsymbol{x}}_i^{r,K} - \boldsymbol{x}_i^{\star}) \nonumber \\
&\hspace{3mm} -\hspace{-0.6mm} \sum_{i=1}^m \boldsymbol{\lambda}_{s|i}^{r+1} (\bar{\boldsymbol{x}}_i^{r, K} - \boldsymbol{x}_i^{\star}) \nonumber \\
&\hspace{3mm} +\sum_{i=1}^m \rho(\boldsymbol{x}_s^{r+1} - \bar{\boldsymbol{x}}_i^{r, K} )^T(\bar{\boldsymbol{x}}_i^{r,K} - \boldsymbol{x}_i^{\star}).
\label{equ:client_inequality_current_sum1}
\end{align}

Next we derive the 2nd expression for $\sum_{i=1}^m (  \bar{\boldsymbol{x}}_i^{r,K} - \boldsymbol{x}_i^{\star} )^T  \boldsymbol{\lambda}_{i|s}^{r+1}$.   To do so, we note that  $\bar{\boldsymbol{x}}_i^{r,K}$ can be represented in terms of $\boldsymbol{\lambda}_{i|s}^{r+1}$ as
\begin{align}
\hspace{-3mm} \bar{\boldsymbol{x}}_i^{r, K} \hspace{-0.5mm}=\hspace{-0.5mm} \boldsymbol{x}_s^{r} \hspace{-0.5mm}-\hspace{-0.5mm} \frac{1}{\rho } (\boldsymbol{\lambda}_{s|i}^r \hspace{-0.5mm}+\hspace{-0.5mm} \boldsymbol{\lambda}_{i | s}^{r+1})  \;\;  i=1,\ldots, m, \label{equ:lambda_update_reverse}
\end{align}
Similarly to the derivation for (\ref{equ:client_inequality_current_sum1}),  we plug the expression (\ref{equ:lambda_update_reverse}) for $\bar{\boldsymbol{x}}_i^{r,K}$ where appropriate, which is given by 
\begin{align}
&\sum_{i=1}^m (  \bar{\boldsymbol{x}}_i^{r,K} - \boldsymbol{x}_i^{\star} )^T  \boldsymbol{\lambda}_{i|s}^{r+1} \nonumber \\
&=\hspace{-0.7mm} \sum_{i=1}^m (\boldsymbol{\lambda}_{i|s}^{r+1} \hspace{-0.7mm}-\hspace{-0.7mm} \boldsymbol{\lambda}_{i | s }^{\star}  )^T \hspace{-0.5mm}  \bar{\boldsymbol{x}}_i^{r,K} \hspace{-0.7mm}-\hspace{-0.7mm} \sum_{i=1}^m \hspace{-0.7mm} \boldsymbol{\lambda}_{i | s}^{r+1,T}\hspace{-0.5mm}    \boldsymbol{x}_i^{\star}  + \sum_{i=1}^m \boldsymbol{\lambda}_{i | s}^{\star,T}  \bar{\boldsymbol{x}}_i^{r,K}   \nonumber 
\end{align}
\begin{align}
&= \sum_{i=1}^m \left [ \boldsymbol{x}_s^{r} \hspace{-0.5mm}-\hspace{-0.5mm} \frac{1}{\rho} (\boldsymbol{\lambda}_{s|i}^r \hspace{-0.5mm}+\hspace{-0.5mm} \boldsymbol{\lambda}_{i | s}^{r+1}) \right]^T (\boldsymbol{\lambda}_{i|s}^{r+1} -\boldsymbol{\lambda}_{i|s}^{\star}  )   \nonumber \\
&\hspace{3mm} -\hspace{-0.7mm} \sum_{i=1}^m \hspace{-0.7mm} \boldsymbol{\lambda}_{i | s}^{r+1,T}\hspace{-0.5mm}  \boldsymbol{x}_i^{\star} + \sum_{i=1}^m  \boldsymbol{\lambda}_{i | s}^{\star,T}   \bar{\boldsymbol{x}}_i^{r,K}    \nonumber \\
&= \sum_{i=1}^m  \left [(\boldsymbol{x}_s^{r} -\boldsymbol{x}_s^{r+1} ) \hspace{-0.5mm}-\hspace{-0.5mm} \frac{1}{\rho}(\boldsymbol{\lambda}_{s|i}^r \hspace{-0.5mm}+\hspace{-0.5mm} \boldsymbol{\lambda}_{i | s}^{r+1}) \right]^T (\boldsymbol{\lambda}_{i|s}^{r+1} -\boldsymbol{\lambda}_{i|s}^{\star} ) \hspace{-0.5mm} \nonumber \\
&+\hspace{-0.5mm}  \sum_{i =1}^m  \boldsymbol{x}_s^{r+1,T}  (\boldsymbol{\lambda}_{i|s}^{r+1} \hspace{-0.5mm}-\hspace{-0.5mm} \boldsymbol{\lambda}_{i|s}^{\star} )    -\hspace{-0.7mm} \sum_{i=1}^m \hspace{-0.7mm} \boldsymbol{\lambda}_{i | s}^{r+1,T}\hspace{-0.5mm}   \boldsymbol{x}_i^{\star}  + \sum_{i=1}^m  \boldsymbol{\lambda}_{i | s}^{\star,T}     \bar{\boldsymbol{x}}_i^{r,K} \nonumber \\
&= \sum_{i=1}^m   \left [(\boldsymbol{x}_s^{r} -\boldsymbol{x}_s^{r+1} ) \hspace{-0.5mm}-\hspace{-0.5mm} \frac{1}{\rho}(\boldsymbol{\lambda}_{s|i}^r \hspace{-0.5mm}-\hspace{-0.5mm} \boldsymbol{\lambda}_{s|i}^{r+1}) \right]^T (\boldsymbol{\lambda}_{i|s}^{r+1} -\boldsymbol{\lambda}_{i|s}^{\star} ) \hspace{-0.5mm} \nonumber \\
&+\hspace{-0.5mm}  \sum_{i =1}^m  \boldsymbol{x}_s^{r+1,T}  (\boldsymbol{\lambda}_{i|s}^{r+1} \hspace{-0.5mm}-\hspace{-0.5mm} \boldsymbol{\lambda}_{i|s}^{\star} ) -\hspace{-0.7mm} \sum_{i=1}^m \hspace{-0.7mm} \boldsymbol{\lambda}_{i | s}^{r+1,T}\hspace{-0.5mm}   \boldsymbol{x}_i^{\star}  + \sum_{i=1}^m  \boldsymbol{\lambda}_{i | s}^{\star,T}     \bar{\boldsymbol{x}}_i^{r,K}  \nonumber \\
& \hspace{-0mm} -  \sum_{i=1}^m \frac{1}{\rho}(\boldsymbol{\lambda}_{s|i}^{r+1} \hspace{-0.5mm}+\hspace{-0.5mm} \boldsymbol{\lambda}_{i | s}^{r+1})^T(\boldsymbol{\lambda}_{i|s}^{r+1} -\boldsymbol{\lambda}_{i|s}^{\star} ).
\label{equ:client_inequality_current_sum2}
\end{align}

In the 2nd step, we derive two different but mathematically equivalent expressions for $ \sum_{i=1}^m \boldsymbol{\lambda}_{s|i}^{r+1,T}( \boldsymbol{x}_s^{r+1} - \boldsymbol{x}_s^{\star}) $.  By using  (\ref{equ:s_lambda_equality}) and the expression for $\boldsymbol{\lambda}_{s|i}^{r+1} = \rho(\bar{\boldsymbol{x}}_i^{r, K} - \boldsymbol{x}_s^{r+1} ) -  \boldsymbol{\lambda}_{i|s}^{r+1} $, we have 
\begin{align}
0 &= \sum_{i=1}^m \boldsymbol{\lambda}_{s|i}^{r+1,T}( \boldsymbol{x}_s^{r+1} - \boldsymbol{x}_s^{\star})  \nonumber \\
 &  = \sum_{i=1}^m \left[\rho (\bar{\boldsymbol{x}}_i^{r,K} - \boldsymbol{x}_s^{r+1} ) 
- \boldsymbol{\lambda}_{i|s}^{r+1} \right]^{T}( \boldsymbol{x}_s^{r+1} - \boldsymbol{x}_s^{\star})  \nonumber \\
 &  = \sum_{i=1}^m \rho(\bar{\boldsymbol{x}}_i^{r, K} - \boldsymbol{x}_s^{r+1} )^T( \boldsymbol{x}_s^{r+1} -   \boldsymbol{x}_s^{\star})  
\nonumber \\
&- \sum_{i=1}^m \boldsymbol{\lambda}_{i|s}^{r+1,T}( \boldsymbol{x}_s^{r+1} - \boldsymbol{x}_s^{\star}).
\label{equ:server_inequality_current_sum1} 
 \end{align}
The 2nd expression for $ \sum_{i=1}^m \boldsymbol{\lambda}_{s|i}^{r+1,T}( \boldsymbol{x}_s^{r+1} - \boldsymbol{x}_s^{\star}) $ can be derived by utilising $\boldsymbol{x}_s^{r+1} = \bar{\boldsymbol{x}}_i^{r,K}-\frac{1}{\rho}(\boldsymbol{\lambda}_{s|i}^{r+1} + \boldsymbol{\lambda}_{i|s}^{r+1})$ as: 
\begin{align}
\hspace{-2mm}0 &= \sum_{i=1}^m \boldsymbol{\lambda}_{s|i}^{r+1,T} ( \boldsymbol{x}_s^{r+1} - \boldsymbol{x}_s^{\star})  \nonumber \\
&= \hspace{-0.7mm}  \sum_{i=1}^m \left( \boldsymbol{\lambda}_{s|i}^{r+1}  \hspace{-0.7mm}- \hspace{-0.7mm}  \boldsymbol{\lambda}_{s|i}^{\star} \right)^T  \hspace{-0.7mm}\boldsymbol{x}_s^{r+1} 
 \hspace{-0.7mm}+ \hspace{-0.7mm}\sum_{i=1}^m  \hspace{-0.7mm}  \boldsymbol{\lambda}_{s|i}^{\star,T} \boldsymbol{x}_s^{r+1}
  \hspace{-0.7mm}- \hspace{-0.7mm}  \sum_{i=1}^m  \hspace{-0.7mm} \boldsymbol{\lambda}_{s|i}^{r+1,T}  \boldsymbol{x}_s^{\star} \nonumber \\
&=  \sum_{i=1}^m \left( \boldsymbol{\lambda}_{s|i}^{r+1} -  \boldsymbol{\lambda}_{s|i}^{\star} \right)^T \left[\bar{\boldsymbol{x}}_i^{r, K}-\frac{1}{\rho}(\boldsymbol{\lambda}_{s|i}^{r+1} + \boldsymbol{\lambda}_{i|s}^{r+1}) \right]  \nonumber\\
&+\sum_{i=1}^m  \boldsymbol{\lambda}_{s|i}^{\star,T} \boldsymbol{x}_s^{r+1}
 - \sum_{i=1}^m \boldsymbol{\lambda}_{s|i}^{r+1,T}  \boldsymbol{x}_s^{\star} \nonumber \\ 
&=  \sum_{i=1}^m \left( \boldsymbol{\lambda}_{s|i}^{r+1} -  \boldsymbol{\lambda}_{s|i}^{\star} \right)^T  \bar{\boldsymbol{x}}_i^{r,K} 
+\sum_{i=1}^m  \boldsymbol{\lambda}_{s|i}^{\star,T} \boldsymbol{x}_s^{r+1}
 \nonumber \\
&- \hspace{-0.7mm} \sum_{i=1}^m  \hspace{-0.7mm} \left( \boldsymbol{\lambda}_{s|i}^{r+1}  \hspace{-0.7mm}- \hspace{-0.7mm}  \boldsymbol{\lambda}_{s|i}^{\star} \right)^T\frac{1}{\rho}(\boldsymbol{\lambda}_{s|i}^{r+1}  \hspace{-0.7mm}+ \hspace{-0.7mm} \boldsymbol{\lambda}_{i|s}^{r+1})    \hspace{-0.7mm}- \hspace{-0.7mm} \sum_{i=1}^m \boldsymbol{\lambda}_{s|i}^{r+1,T}  \boldsymbol{x}_s^{\star}. \hspace{-2mm}
 \label{equ:server_inequality_current_sum2}  
 \end{align}

Finally, combining (\ref{equ:client_inequality_current_sum1}) and (\ref{equ:client_inequality_current_sum2})-(\ref{equ:server_inequality_current_sum2}) produces
\begin{align}
 &2\sum_{i=1}^m (  \bar{\boldsymbol{x}}_i^{r, K} - \boldsymbol{x}_i^{\star} )^T  \boldsymbol{\lambda}_{i|s}^{r+1} \nonumber 
 \end{align}
 \begin{align}
&= \sum_{i=1}^m \Big( \rho( \boldsymbol{x}_s^r - \boldsymbol{x}_s^{r+1} ) 
\hspace{-0.8mm}+\hspace{-0.8mm}\boldsymbol{\lambda}_{s|i}^{r+1}\hspace{-0.7mm}-\hspace{-0.7mm}\boldsymbol{\lambda}_{s|i}^r\Big)^T  (\bar{\boldsymbol{x}}_i^{r, K} \hspace{-0.8mm}-\hspace{-0.8mm} \boldsymbol{x}_i^{\star}) \nonumber \\
&-\hspace{-0.6mm} \sum_{i=1}^m \boldsymbol{\lambda}_{s|i}^{r+1} (\bar{\boldsymbol{x}}_i^{r, K} \hspace{-0.8mm}-\hspace{-0.8mm} \boldsymbol{x}_i^{\star})  \hspace{-0.8mm}+\hspace{-0.8mm}\sum_{i=1}^m \rho(\boldsymbol{x}_s^{r+1} \hspace{-0.8mm}-\hspace{-0.8mm} \bar{\boldsymbol{x}}_i^{r,K} )^T(\bar{\boldsymbol{x}}_i^{r, K} \hspace{-0.8mm}-\hspace{-0.8mm} \boldsymbol{x}_i^{\star}) \nonumber \\
&\hspace{0mm} +\sum_{i=1}^m   \left [(\boldsymbol{x}_s^{r} -\boldsymbol{x}_s^{r+1} ) \hspace{-0.5mm}-\hspace{-0.5mm} \frac{1}{\rho}(\boldsymbol{\lambda}_{s|i}^r \hspace{-0.5mm}-\hspace{-0.5mm} \boldsymbol{\lambda}_{s|i}^{r+1}) \right]^T (\boldsymbol{\lambda}_{i|s}^{r+1}  \hspace{-0.7mm} -\hspace{-0.7mm}  \boldsymbol{\lambda}_{i|s}^{\star} ) \hspace{-0.5mm}\nonumber \\
&+\hspace{-0.5mm}  \sum_{i =1}^m  \boldsymbol{x}_s^{r+1,T}  (\boldsymbol{\lambda}_{i|s}^{r+1} \hspace{-0.7mm}-\hspace{-0.7mm} \boldsymbol{\lambda}_{i|s}^{\star} )  \hspace{-0.7mm} -\hspace{-0.7mm} \sum_{i=1}^m \hspace{-0.7mm} \boldsymbol{\lambda}_{i | s}^{r+1,T}\hspace{-0.5mm}   \boldsymbol{x}_i^{\star}  \hspace{-0.7mm} +\hspace{-0.7mm}  \sum_{i=1}^m  \boldsymbol{\lambda}_{i | s}^{\star,T}  \bar{\boldsymbol{x}}_i^{r, K}   \nonumber \\
& -  \sum_{i=1}^m \frac{1}{\rho}(\boldsymbol{\lambda}_{s|i}^{r+1} \hspace{-0.5mm}+\hspace{-0.5mm} \boldsymbol{\lambda}_{i | s}^{r+1})^T(\boldsymbol{\lambda}_{i|s}^{r+1} -\boldsymbol{\lambda}_{i|s}^{\star} ) \nonumber \\
& \hspace{0mm}  + \sum_{i=1}^m \rho(\bar{\boldsymbol{x}}_i^{r, K}  \hspace{0.7mm}- \hspace{0.7mm} \boldsymbol{x}_s^{r+1} )^T( \boldsymbol{x}_s^{r+1}  \hspace{0.7mm}- \hspace{0.7mm} \boldsymbol{x}_s^{\star})  \nonumber \\
&\hspace{0.7mm} - \hspace{0.7mm} \sum_{i=1}^m \boldsymbol{\lambda}_{i|s}^{r+1,T}( \boldsymbol{x}_s^{r+1} - \boldsymbol{x}_s^{\star})  + \sum_{i=1}^m \left( \boldsymbol{\lambda}_{s|i}^{r+1} -  \boldsymbol{\lambda}_{s|i}^{\star} \right)^T  \bar{\boldsymbol{x}}_i^{r, K} \nonumber \\
&- \sum_{i=1}^m \left( \boldsymbol{\lambda}_{s|i}^{r+1} -  \boldsymbol{\lambda}_{s|i}^{\star} \right)^T\frac{1}{\rho}(\boldsymbol{\lambda}_{s|i}^{r+1} + \boldsymbol{\lambda}_{i|s}^{r+1})  
+\sum_{i=1}^m  \boldsymbol{\lambda}_{s|i}^{\star,T} \boldsymbol{x}_s^{r+1}  \nonumber \\
& - \sum_{i=1}^m \boldsymbol{\lambda}_{s|i}^{r+1,T}  \boldsymbol{x}_s^{\star}  \nonumber \\
& \stackrel{(a)}{=}  \sum_{i=1}^m \frac{1}{\rho} \Big[ \rho( \boldsymbol{x}_s^r - \boldsymbol{x}_s^{r+1} ) 
+\boldsymbol{\lambda}_{s|i}^{r+1}-\boldsymbol{\lambda}_{s|i}^r\Big]^T \nonumber \\
& \hspace{3mm} \cdot [ \rho(\bar{\boldsymbol{x}}_i^{r,K} - \boldsymbol{x}_i^{\star}) + \boldsymbol{\lambda}_{i|s}^{r+1} -\boldsymbol{\lambda}_{i|s}^{\star}]  + 2\sum_{i=1}^m \boldsymbol{\lambda}_{i|s}^{\star} \bar{\boldsymbol{x}}_i^{r, K} \nonumber \\
& \hspace{3mm} - \sum_{i=1}^m \rho \| \boldsymbol{x}_s^{r+1} - \bar{\boldsymbol{x}}_i^{r,K} \|^2 - \sum_{i=1}^m\frac{1}{\rho}\|\boldsymbol{\lambda}_{s|i}^{r+1} + \boldsymbol{\lambda}_{i|s}^{r+1}  \|^2 \nonumber \\
& \stackrel{(b)}{=}  \sum_{i=1}^m \frac{1}{2\rho}  \|\rho(\boldsymbol{x}_s^r - \boldsymbol{x}_i^{\star}) - (\boldsymbol{\lambda}_{s|i}^r +\boldsymbol{\lambda}_{i|s}^{\star} )   \|^2 \nonumber \\
& \hspace{3mm} -  \sum_{i=1}^m \frac{1}{2\rho}  \|\rho(\boldsymbol{x}_s^{r+1} - \boldsymbol{x}_i^{\star}) - (\boldsymbol{\lambda}_{s|i}^{r+1} +\boldsymbol{\lambda}_{i|s}^{\star} )   \|^2 \nonumber\\
&  \hspace{3mm}  - \sum_{i=1}^m \frac{1}{2\rho}  \|\rho(\boldsymbol{x}_s^r - \bar{\boldsymbol{x}}_i^{r, K}) - (\boldsymbol{\lambda}_{s|i}^r +\boldsymbol{\lambda}_{i|s}^{r+1} )   \|^2  \nonumber \\
&  \hspace{3mm}  +  \sum_{i=1}^m \frac{1}{2\rho}  \|\rho(\boldsymbol{x}_s^{r+1} - \bar{\boldsymbol{x}}_i^{r, K}) - (\boldsymbol{\lambda}_{s|i}^{r+1} +\boldsymbol{\lambda}_{i|s}^{r+1} )   \|^2 \nonumber\\
&  \hspace{3mm} + 2\sum_{i=1}^m \boldsymbol{\lambda}_{i|s}^{\star} \bar{\boldsymbol{x}}_i^{r, K} - \sum_{i=1}^m \rho \| \boldsymbol{x}_s^{r+1} - \bar{\boldsymbol{x}}_i^{r, K} \|^2 \nonumber \\
&  \hspace{3mm}  - \sum_{i=1}^m\frac{1}{\rho}\|\boldsymbol{\lambda}_{s|i}^{r+1} + \boldsymbol{\lambda}_{i|s}^{r+1}  \|^2 \nonumber \\
& \stackrel{(c)}{=}  \sum_{i=1}^m \frac{1}{2\rho}  \|\rho(\boldsymbol{x}_s^r - \boldsymbol{x}_i^{\star}) - (\boldsymbol{\lambda}_{s|i}^r +\boldsymbol{\lambda}_{i|s}^{\star} )   \|^2  \nonumber \\
& \hspace{3mm}  - \hspace{-0.7mm}  \sum_{i=1}^m \frac{1}{2\rho}  \|\rho(\boldsymbol{x}_s^{r+1} \hspace{-0.7mm} -\hspace{-0.7mm}  \boldsymbol{x}_i^{\star}) 
\hspace{-0.7mm} -\hspace{-0.7mm}  (\boldsymbol{\lambda}_{s|i}^{r+1} \hspace{-0.7mm} +\hspace{-0.7mm} \boldsymbol{\lambda}_{i|s}^{\star} )   \|^2 \hspace{-0.7mm} +\hspace{-0.7mm}  2\sum_{i=1}^m \boldsymbol{\lambda}_{i|s}^{\star} \bar{\boldsymbol{x}}_i^{r, K}  \nonumber \\
& \stackrel{(d)}{=}  \sum_{i=1}^m \frac{1}{2\rho}  \|\rho(\bar{\boldsymbol{x}}_i^{r, K} - \boldsymbol{x}_i^{\star}) + (\boldsymbol{\lambda}_{i|s}^{r+1} -\boldsymbol{\lambda}_{i|s}^{\star} )   \|^2  \nonumber 
\end{align}
\begin{align}
& \hspace{3mm}  -   \sum_{i=1}^m \frac{1}{2\rho}  \|\rho(\bar{\boldsymbol{x}}_i^{r+1, K} - \boldsymbol{x}_i^{\star}) + (\boldsymbol{\lambda}_{i|s}^{r+2} -\boldsymbol{\lambda}_{i|s}^{\star} )   \|^2 \nonumber \\
& \hspace{3mm} + 2\sum_{i=1}^m \boldsymbol{\lambda}_{i|s}^{\star} \bar{\boldsymbol{x}}_i^{r, K},
\end{align}
where step $(a)$ uses the fact that $\sum_{i=1}^m \boldsymbol{\lambda}_{s|i}^{\star} = \sum_{i=1}^m \boldsymbol{\lambda}_{i|s}^{\star}  = 0$, step $(b)$ follows from Lemma~\ref{lemma:identity},  step $(c)$ uses the identities of $\rho(\boldsymbol{x}_s^r-\bar{\boldsymbol{x}}_i^{r, K}) - (\boldsymbol{\lambda}_{s|i}^r + \boldsymbol{\lambda}_{i|s}^{r+1}) =0 $ and $\rho(\boldsymbol{x}_s^{r+1}-\bar{\boldsymbol{x}}_i^{r, K}) + (\boldsymbol{\lambda}_{s|i}^{r+1} + \boldsymbol{\lambda}_{i|s}^{r+1}) =0 $ from (\ref{equ:client_update_split})-(\ref{equ:server_update_split}), and step $(d)$ uses $\rho(\boldsymbol{x}_s^k-\bar{\boldsymbol{x}}_i^{r,K}) - (\boldsymbol{\lambda}_{s|i}^r + \boldsymbol{\lambda}_{i|s}^{r+1}) =0 $ and $\rho(\boldsymbol{x}_s^{r+1}-\bar{\boldsymbol{x}}_i^{r+1, K}) - (\boldsymbol{\lambda}_{s|i}^{r+1} + \boldsymbol{\lambda}_{i|s}^{r+2}) =0 $. The proof is complete. 
\end{proof}

\section{Proof for Lemma~\ref{lemma:lower_bound} }
\label{appendix:lemma_lowerbound}
\begin{proof}
The lower bound in (\ref{equ:lowerbound}) can be easily proved to be:
\begin{align}
& \sum_{i=1}^m  \Big[   f_i(\boldsymbol{x}_i ) - f_i(\boldsymbol{x}_i^{\star}) - \boldsymbol{x}_i^{T}\boldsymbol{\lambda}_{i|s}^{\star} \Big] \nonumber \\
&\stackrel{(a)}{\geq} \sum_{i=1}^m  \Big[    -f_i^{\ast}(\boldsymbol{\lambda}_{i|s}^{\star}) - f_i(\boldsymbol{x}_i^{\star})  \Big] \stackrel{}{=} 0 \nonumber,
\end{align}
where $f_i^{\ast}(\cdot)$ is the conjugate function of $f_i(\cdot)$ as defined in (\ref{equ:conj_def}). Step~$(a)$ uses Fenchel's inequality (see \cite{Boyd04ConvexOptimization}).  It is known that for a convex function, the duality gap is 0 at the optimal solution. The proof is complete. 
\end{proof}

\section{Proof for Theorem~\ref{theorem:linear_conv}}
\label{appendix:linear_conv}

\begin{proof}
The proof for Theorem~\ref{theorem:linear_conv} is mainly based on the results in Lemma~\ref{lemma:twoBounds} and \ref{lemma:lower_bound}. Assume that $1>\theta>0$ and $1/\eta > L\geq \mu >0 \}$. The RHS of (\ref{equ:upper_bound_final}) in Lemma~\ref{lemma:twoBounds} can be further lower bounded by
\begin{align}
 & \sum_{i=1}^m  \frac{1}{K}\sum_{k=0}^{K-1} \hspace{-0.6mm} \frac{1/\eta - \theta \mu}{2}\|\boldsymbol{x}_i^{r,k}-\boldsymbol{x}_i^{\star} \|^2 \hspace{-0.6mm} \nonumber \\
  &+  \sum_{i=1}^m  \frac{1}{4\rho}  \|\rho(\bar{\boldsymbol{x}}_i^{r, K} - \boldsymbol{x}_i^{\star}) + (\boldsymbol{\lambda}_{i|s}^{r+1} -\boldsymbol{\lambda}_{i|s}^{\star} )   \|^2 \nonumber \\
  \hspace{-2mm}&\geq \sum_{i=1}^m \Big[ f_i(\bar{\boldsymbol{x}}_i^{r, K} ) - (\bar{\boldsymbol{x}}_i^{r,K})^T \boldsymbol{\lambda}_{i|s}^{\star} - f_i(\boldsymbol{x}_i^{\star})  \nonumber \\
  & \hspace{9mm} + \frac{1}{K}\sum_{k=0}^{K-1} \Big(\frac{1}{2\eta}  \|\boldsymbol{x}_i^{\star} - \boldsymbol{x}_i^{r, k+1} \|^2  \nonumber\\
\hspace{-3mm}&\hspace{9mm} +\hspace{-0.6mm} \frac{1/\eta - L}{2} \|\boldsymbol{x}_i^{r, k+1}  \hspace{-0.6mm}- \hspace{-0.6mm}\boldsymbol{x}_i^{r, k} \|^2 \hspace{-0.7mm}  + \hspace{-0.7mm}  \frac{1-\theta}{2L}\|  \rho(\hspace{-0.6mm} \boldsymbol{x}_s^{r} \hspace{-0.7mm} - \hspace{-0.6mm}  \boldsymbol{x}_i^{r,k+1} ) \hspace{-0.6mm}\nonumber \\
\hspace{-3mm}&  \hspace{9mm}  -\hspace{-0.6mm}  \boldsymbol{\lambda}_{s|i}^{r} \hspace{-0.6mm} - \boldsymbol{\lambda}_{i|s}^{\star} -\hspace{-0.6mm} (1/\eta)(\boldsymbol{x}_i^{r, k+1} \hspace{-0.6mm} -\hspace{-0.6mm}  \boldsymbol{x}_i^{r, k}) \hspace{-0.6mm}\|^2 \Big)\nonumber \\
 &\hspace{9mm} +  \frac{1}{4\rho}  \|\rho(\bar{\boldsymbol{x}}_i^{r+1, K} - \boldsymbol{x}_i^{\star}) + (\boldsymbol{\lambda}_{i|s}^{r+2} -\boldsymbol{\lambda}_{i|s}^{\star} )   \|^2 \Big]  \nonumber \\
&\stackrel{(a)}{\geq} \sum_{i=1}^m \Big[ \frac{1}{K}\sum_{k=0}^{K-1} \Big(\frac{1}{2\eta}  \|\boldsymbol{x}_i^{\star} - \boldsymbol{x}_i^{r, k+1} \|^2  \nonumber\\
\hspace{-3mm}&\hspace{9mm} +\hspace{-0.6mm} \frac{1/\eta - L}{2} \|\boldsymbol{x}_i^{r, k+1}  \hspace{-0.6mm}- \hspace{-0.6mm}\boldsymbol{x}_i^{r, k} \|^2 \hspace{-0.7mm}  + \hspace{-0.7mm}  \frac{1-\theta}{2\eta^2 L}\|  \eta( \rho(\hspace{-0.6mm} \boldsymbol{x}_s^{r} \hspace{-0.7mm} - \hspace{-0.6mm}  \boldsymbol{x}_i^{r,k+1} ) \hspace{-0.6mm}\nonumber 
\end{align}
\begin{align}
\hspace{-3mm}&  \hspace{9mm}  -\hspace{-0.6mm}  \boldsymbol{\lambda}_{s|i}^{r} \hspace{-0.6mm} - \boldsymbol{\lambda}_{i|s}^{\star}) -\hspace{-0.6mm} (\boldsymbol{x}_i^{r, k+1} \hspace{-0.6mm} -\hspace{-0.6mm}  \boldsymbol{x}_i^{r, k}) \hspace{-0.6mm}\|^2  \Big) \nonumber \\
 &\hspace{5mm} +  \frac{1}{4\rho}  \|\rho(\bar{\boldsymbol{x}}_i^{r+1, K} - \boldsymbol{x}_i^{\star}) + (\boldsymbol{\lambda}_{i|s}^{r+2} -\boldsymbol{\lambda}_{i|s}^{\star} )   \|^2 \Big] \nonumber \\
 &\stackrel{(b)}{\geq} \sum_{i=1}^m \Big[ \frac{1}{K}\sum_{k=0}^{K-1} \Big(  \frac{1/\eta -\theta\mu \phi + \theta\mu \phi}{2} \|\boldsymbol{x}_i^{\star} - \boldsymbol{x}_i^{r, k+1} \|^2  \nonumber\\
\hspace{-3mm}&\hspace{9mm} + \hspace{-0.7mm}  \frac{\gamma_1}{2} \|  \eta( \rho(\hspace{-0.6mm} \boldsymbol{x}_s^{r} \hspace{-0.7mm} - \hspace{-0.6mm}  \boldsymbol{x}_i^{r,k+1} )   -\hspace{-0.6mm}  \boldsymbol{\lambda}_{s|i}^{r} \hspace{-0.6mm} - \boldsymbol{\lambda}_{i|s}^{\star}) \hspace{-0.6mm}\|^2\Big)  \nonumber \\
 &\hspace{5mm} +  \frac{1}{4\rho}  \|\rho(\bar{\boldsymbol{x}}_i^{r+1, K} - \boldsymbol{x}_i^{\star}) + (\boldsymbol{\lambda}_{i|s}^{r+2} -\boldsymbol{\lambda}_{i|s}^{\star} )   \|^2 \Big] \nonumber \\
  &\stackrel{(c)}{\geq} \sum_{i=1}^m \Big[ \frac{1}{K}\sum_{k=0}^{K-1}  \frac{1/\eta -\theta\mu \phi}{2}  \|\boldsymbol{x}_i^{\star} - \boldsymbol{x}_i^{r, k+1} \|^2  \nonumber\\
\hspace{-3mm}&\hspace{9mm} + \frac{\theta\mu \phi}{2}  \|\boldsymbol{x}_i^{\star} - \bar{\boldsymbol{x}}_i^{r, K} \|^2 + \hspace{-0.7mm}  \frac{\gamma_1 \eta^2}{2} \|   \boldsymbol{\lambda}_{i|s}^{r+1} \hspace{-0.6mm} - \boldsymbol{\lambda}_{i|s}^{\star} \hspace{-0.6mm}\|^2  \nonumber \\
 &\hspace{5mm} +  \frac{1}{4\rho}  \|\rho(\bar{\boldsymbol{x}}_i^{r+1, K} - \boldsymbol{x}_i^{\star}) + (\boldsymbol{\lambda}_{i|s}^{r+2} -\boldsymbol{\lambda}_{i|s}^{\star} )   \|^2 \Big] \nonumber \\
&\stackrel{(d)}{\geq} \sum_{i=1}^m \Big[ \frac{1}{K}\sum_{k=0}^{K-1}  \frac{1/\eta -\theta\mu \phi}{2}  \|\boldsymbol{x}_i^{\star} - \boldsymbol{x}_i^{r, k+1} \|^2  \nonumber\\
\hspace{-3mm}&\hspace{9mm} +\frac{\gamma_2}{2}  \|\rho(\bar{\boldsymbol{x}}_i^{r, K} - \boldsymbol{x}_i^{\star}) + (\boldsymbol{\lambda}_{i|s}^{r+1} -\boldsymbol{\lambda}_{i|s}^{\star} )   \|^2  \nonumber \\
 &\hspace{5mm} +  \frac{1}{4\rho}  \|\rho(\bar{\boldsymbol{x}}_i^{r+1, K} - \boldsymbol{x}_i^{\star}) + (\boldsymbol{\lambda}_{i|s}^{r+2} -\boldsymbol{\lambda}_{i|s}^{\star} )   \|^2 \Big], \label{equ:proof_theoremLinear_1} 
\end{align}
where step $(a)$ follows from Lemma~\ref{lemma:lower_bound}. Step $(b)$ introduces $1>\phi>0$ and utilises the inequality $\|\boldsymbol{b}\|^2+\|\boldsymbol{c}\|^2\geq \frac{1}{2}\|\boldsymbol{b}+ \boldsymbol{c}\|^2$. The parameter $\gamma_{1}$ is defined as  
\begin{align}
\gamma_{1} &= \min\left(\frac{1-\theta}{2L \eta^2} , \frac{1/\eta - L}{2}\right).
\label{equ:gamma_i1}
\end{align}
Step $(c)$ employs Jensen's inequality and $\boldsymbol{\lambda}_{i|s}^{r+1} = \rho(\boldsymbol{x}_s^r - \bar{\boldsymbol{x}}_i^{r,K}) - \boldsymbol{\lambda}_{s|i}^r$. Step $(d)$ utilises the inequality $\|\boldsymbol{b}\|^2+\|\boldsymbol{c}\|^2\geq \frac{1}{2}\|\boldsymbol{b}+ \boldsymbol{c}\|^2$ again, and the parameter $\gamma_{2}$ is defined as  
\begin{align}
\gamma_{2} &= \min\left( \frac{\theta\mu\phi}{2\rho^2},  \frac{\gamma_{1}\eta^2}{2}  \right).
\label{equ:gamma_i2}
\end{align}

By using  $\{\boldsymbol{x}_i^{r-1,K}= \boldsymbol{x}_i^{r,k=0}\}$, the inequality (\ref{equ:proof_theoremLinear_1}) can be reformulated as
\begin{align}
 & \sum_{i=1}^m   \hspace{-0.6mm} \frac{1/\eta - \theta \mu}{2K}\|\boldsymbol{x}_i^{r-1,K}-\boldsymbol{x}_i^{\star} \|^2 \hspace{-0.6mm} \nonumber \\
  &+  \sum_{i=1}^m  \Big(\frac{1}{4\rho} - \frac{\gamma_2}{2} \Big) \|\rho(\bar{\boldsymbol{x}}_i^{r, K} - \boldsymbol{x}_i^{\star}) + (\boldsymbol{\lambda}_{i|s}^{r+1} -\boldsymbol{\lambda}_{i|s}^{\star} )   \|^2 \nonumber \\
 & \geq \sum_{i=1}^m \frac{1}{K}\sum_{k=1}^{K-1}  \frac{\theta\mu (1-\phi)}{2}  \|\boldsymbol{x}_i^{\star} - \boldsymbol{x}_i^{r, k} \|^2  \nonumber\\
& + \sum_{i=1}^m  \frac{1/\eta -\theta\mu \phi}{2K}  \|\boldsymbol{x}_i^{\star} - \boldsymbol{x}_i^{r, K} \|^2  \nonumber\\
 &\hspace{5mm} + \sum_{i=1}^m  \frac{1}{4\rho}  \|\rho(\bar{\boldsymbol{x}}_i^{r+1, K} - \boldsymbol{x}_i^{\star}) + (\boldsymbol{\lambda}_{i|s}^{r+2} -\boldsymbol{\lambda}_{i|s}^{\star} )   \|^2 \nonumber \\
 & \geq \sum_{i=1}^m  \frac{1/\eta -\theta\mu \phi}{2K}  \|\boldsymbol{x}_i^{\star} - \boldsymbol{x}_i^{r, K} \|^2  \nonumber 
 \end{align}
 \begin{align}
 &\hspace{5mm} + \sum_{i=1}^m  \frac{1}{4\rho}  \|\rho(\bar{\boldsymbol{x}}_i^{r+1, K} - \boldsymbol{x}_i^{\star}) + (\boldsymbol{\lambda}_{i|s}^{r+2} -\boldsymbol{\lambda}_{i|s}^{\star} )   \|^2.
\label{equ:proof_theoremLinear_2}
\end{align}
%\begin{align}
% & \sum_{i=1}^m  \frac{1}{K}\sum_{k=0}^{K-1} \hspace{-0.6mm} \frac{1/\eta - \theta \mu}{2}\|\boldsymbol{x}_i^{r,k}-\boldsymbol{x}_i^{\star} \|^2 \hspace{-0.6mm} \nonumber \\
%  &+  \sum_{i=1}^m  \Big(\frac{1}{4\rho} - \frac{\gamma_2}{2} \Big) \|\rho(\bar{\boldsymbol{x}}_i^{r, K} - \boldsymbol{x}_i^{\star}) + (\boldsymbol{\lambda}_{i|s}^{r+1} -\boldsymbol{\lambda}_{i|s}^{\star} )   \|^2 \nonumber \\
%& \geq \sum_{i=1}^m \frac{1}{K}\sum_{k=0}^{K-1}  \frac{1/\eta -\theta\mu \phi}{2}  \|\boldsymbol{x}_i^{\star} - \boldsymbol{x}_i^{r, k+1} \|^2  \nonumber\\
% &\hspace{5mm} + \sum_{i=1}^m  \frac{1}{4\rho}  \|\rho(\bar{\boldsymbol{x}}_i^{r+1, K} - \boldsymbol{x}_i^{\star}) + (\boldsymbol{\lambda}_{i|s}^{r+2} -\boldsymbol{\lambda}_{i|s}^{\star} )   \|^2.
%\label{equ:proof_theoremLinear_2}
%\end{align}
We note that when $\phi$ is chosen to satisfy $ \frac{1}{4\rho} > \frac{\theta\mu \phi}{4\rho^2}$, we have $ \frac{1}{4\rho} > \frac{\theta\mu \phi}{4\rho^2} \geq \frac{\gamma_{2}}{2} $ based on the definition of $\gamma_{2}$ in (\ref{equ:gamma_i2}). As a result, it is clear from (\ref{equ:proof_theoremLinear_2}) that the coefficients before $\|\rho(\boldsymbol{x}_i^{r-1,K}-\boldsymbol{x}_i^{\star}) \|^2 $  and $\|\rho(\bar{\boldsymbol{x}}_i^{r, K} - \boldsymbol{x}_i^{\star}) + (\boldsymbol{\lambda}_{i|s}^{r+1} -\boldsymbol{\lambda}_{i|s}^{\star} )   \|^2$ are smaller than those coefficients before  $\|\rho(\boldsymbol{x}_i^{r, K}-\boldsymbol{x}_i^{\star}) \|^2 $  and $\|\rho\bar{\boldsymbol{x}}_i^{r+1, K} - \boldsymbol{x}_i^{\star}) + (\boldsymbol{\lambda}_{i|s}^{r+2} -\boldsymbol{\lambda}_{i|s}^{\star} )   \|^2$. Therefore, we can conclude that GPDMM has linear convergence rate under certain conditions. The expression for the parameter $\beta$ in Theorem~\ref{theorem:linear_conv} can be easily derived from (\ref{equ:proof_theoremLinear_2}). The proof is complete.  
\end{proof}
%\begin{align}
%\beta &= \max\left(\max_{i=1}^m \frac{1/(4\rho)-\gamma_{i,2}/2 }{1/(4\rho)}, \max_{i=1}^m \frac{ 1/\eta_i - \theta_i\mu_i }{ 1/\eta_i - \theta_i\mu_i\phi_i}  \right)
%\end{align}

\section{Proof for Theorem~\ref{theorem:sublinear}}
\label{appendix:sublinear}

\begin{proof}
Similar to Appendix~\ref{appendix:linear_conv}, the proof for Theorem~\ref{theorem:sublinear} is also based on the results in Lemma~\ref{lemma:twoBounds} and \ref{lemma:lower_bound}. Summing the inequality (\ref{equ:upper_bound_final}) in Lemma~\ref{lemma:twoBounds} from $r=1$ until $r=R$ and setting $\mu=0$ and $\theta=0$ produces   
\begin{align}
 &\frac{1}{R}\sum_{r=1}^R  \frac{1}{K} \Big[\hspace{-0.6mm} \frac{1/\eta }{2}\|\boldsymbol{x}_i^{0,K}-\boldsymbol{x}_i^{\star} \|^2 \hspace{-0.6mm} \nonumber \\
  &\hspace{8mm} + \frac{1}{4\rho}  \|\rho(\bar{\boldsymbol{x}}_i^{1, K} - \boldsymbol{x}_i^{\star}) + (\boldsymbol{\lambda}_{i|s}^{2} -\boldsymbol{\lambda}_{i|s}^{\star} )   \|^2 \Big] \nonumber \\
  \hspace{-2mm}&\geq \frac{1}{R}\sum_{r=1}^R\sum_{i=1}^m \Big[ f_i(\bar{\boldsymbol{x}}_i^{r, K} ) - (\bar{\boldsymbol{x}}_i^{r,K})^T \boldsymbol{\lambda}_{i|s}^{\star} - f_i(\boldsymbol{x}_i^{\star})  \nonumber \\
\hspace{-1mm}&\hspace{1mm} +\hspace{-0.7mm}\frac{1}{K}\hspace{-0.7mm}\sum_{k=0}^{K-1}\hspace{-0.7mm}\Big(\hspace{-0.6mm} \frac{1/\eta \hspace{-0.7mm}-\hspace{-0.7mm} L}{2} \|\boldsymbol{x}_i^{r, k+1}  \hspace{-0.7mm}- \hspace{-0.7mm}\boldsymbol{x}_i^{r, k} \|^2 \hspace{-0.7mm}  + \hspace{-0.7mm}  \frac{1}{2L\eta^2}\|  \eta(\rho(\hspace{-0.6mm} \boldsymbol{x}_s^{r} \hspace{-0.7mm} - \hspace{-0.6mm}  \boldsymbol{x}_i^{r,k+1} ) \hspace{-0.6mm}\nonumber \\
\hspace{-1mm}&  \hspace{9mm}  -\hspace{-0.6mm}  \boldsymbol{\lambda}_{s|i}^{r} \hspace{-0.6mm} - \boldsymbol{\lambda}_{i|s}^{\star}) -\hspace{-0.6mm} (\boldsymbol{x}_i^{r, k+1} \hspace{-0.6mm} -\hspace{-0.6mm}  \boldsymbol{x}_i^{r, k}) \hspace{-0.6mm}\|^2 \Big) \Big] \nonumber \\
  \hspace{-2mm}&\stackrel{(a)}{\geq} \frac{1}{R}\sum_{r=1}^R\sum_{i=1}^m \Big[ f_i(\bar{\boldsymbol{x}}_i^{r, K} ) - (\bar{\boldsymbol{x}}_i^{r,K})^T \boldsymbol{\lambda}_{i|s}^{\star} - f_i(\boldsymbol{x}_i^{\star})  \nonumber \\
\hspace{-1mm}&\hspace{1mm} +\hspace{-0.7mm}\frac{1}{K}\hspace{-0.7mm}\sum_{k=0}^{K-1}\hspace{-0.7mm}\Big(\hspace{-0.6mm}  \frac{\gamma_1}{2}\|  \eta(\rho(\hspace{-0.6mm} \boldsymbol{x}_s^{r} \hspace{-0.7mm} - \hspace{-0.6mm}  \boldsymbol{x}_i^{r,k+1} ) \hspace{-0.6mm}  - \hspace{-0.6mm}  \boldsymbol{\lambda}_{s|i}^{r} \hspace{-0.6mm} - \boldsymbol{\lambda}_{i|s}^{\star}) \|^2 \Big) \Big] \nonumber \\
  \hspace{-2mm}&\stackrel{(b)}{\geq} \frac{1}{R}\sum_{r=1}^R\sum_{i=1}^m \Big[ f_i(\bar{\boldsymbol{x}}_i^{r, K} ) - (\bar{\boldsymbol{x}}_i^{r,K})^T \boldsymbol{\lambda}_{i|s}^{\star} - f_i(\boldsymbol{x}_i^{\star})  \nonumber \\
\hspace{-1mm}&\hspace{15mm} +\hspace{-0.7mm}\Big(\hspace{-0.6mm}  \frac{\gamma_1\eta^2}{2}\|  \boldsymbol{x}_{i|s}^{r+1}  \hspace{-0.6mm} - \boldsymbol{\lambda}_{i|s}^{\star} \|^2 \Big) \Big] \nonumber \\
  \hspace{-2mm}&\stackrel{(c)}{\geq} \sum_{i=1}^m \Big[ f_i(\bar{\boldsymbol{x}}_i^{R, K} ) - (\bar{\boldsymbol{\lambda}}_i^{R,K})^T \boldsymbol{\lambda}_{i|s}^{\star} - f_i(\boldsymbol{x}_i^{\star})  \nonumber \\
\hspace{-1mm}&\hspace{15mm} +\hspace{-0.7mm}\Big(\hspace{-0.6mm}  \frac{\gamma_1\eta^2}{2}\|  \bar{\boldsymbol{\lambda}}_{i|s}^{R}  \hspace{-0.6mm} - \boldsymbol{\lambda}_{i|s}^{\star} \|^2 \Big) \Big],
\label{appendix:theorem_sublinear_proof_1}
\end{align}
where step $(a)$ utilises the inequality $\|\boldsymbol{b}\|^2+\|\boldsymbol{c}\|^2\geq \frac{1}{2}\|\boldsymbol{b}+ \boldsymbol{c}\|^2$, and the parameter $\gamma_{1}$ is given by (\ref{equ:gamma_i1}) with $\theta=0$. Step $(b)$ employs Jensen's inequality and $\boldsymbol{\lambda}_{i|s}^{r+1} = \rho(\boldsymbol{x}_s^r - \bar{\boldsymbol{x}}_i^{r,K}) - \boldsymbol{\lambda}_{s|i}^r$. Step $(b)$ employs Jensen's inequality again. The results in Theorem \ref{theorem:sublinear} follows directly using the property that the LHS of (\ref{appendix:theorem_sublinear_proof_1}) decays in the order of $O(1/R)$. The proof is complete. \end{proof}

\ifCLASSOPTIONcaptionsoff
  \newpage
\fi

\bibliographystyle{IEEEtran}
\end{document}